\documentclass[11pt]{article}
\usepackage[margin=1in]{geometry}
\usepackage{booktabs}
\usepackage{amsmath, amsfonts, amssymb, amsthm}
\usepackage{thmtools, bm, bbm, thm-restate}
\usepackage{cite}

\usepackage{algorithmicx}
\usepackage[noend]{algpseudocode}
\usepackage{xcolor}
\definecolor{BrickRed}{rgb}{0.8,0.25,0.33}
\usepackage{hyperref}
\hypersetup{colorlinks,citecolor=blue,linkcolor=BrickRed}
\usepackage[hyperpageref]{backref}
\usepackage{cleveref}
\usepackage{crossreftools}
\usepackage{asymptote}
\usepackage{graphicx}
\usepackage{multirow}
\usepackage{comment}
\usepackage{dsfont}
\usepackage{bigstrut}
\usepackage{caption}
\usepackage{subcaption}
\usepackage{tcolorbox}
\usepackage{algorithm}
\usepackage[noend]{algpseudocode}

\usepackage{cancel}
\usepackage{paralist}
\usepackage[framemethod=tikz]{mdframed}
\usepackage{nicefrac}

\usepackage{pifont}

\usepackage{enumitem}

\theoremstyle{plain}
\newtheorem{thm}{Theorem}[section]
\newtheorem{cor}[thm]{Corollary}
\newtheorem{prop}[thm]{Proposition}

\newtheorem{fact}[thm]{Fact}

\newtheorem{lemma}[thm]{Lemma}

\newtheorem{obs}[thm]{Observation}

\newtheorem{exm}[thm]{Example}

\theoremstyle{definition}

\crefname{thm}{Theorem}{theorems}
\crefname{cla}{Claim}{claims}
\crefname{lem}{Lemma}{lemmas}
\crefname{fact}{Fact}{facts}

\newcommand{\approxconstant}{0.0115}

\newcommand{\e}{\mathrm{e}}
\newcommand{\E}{\mathbb{E}}

\newcommand{\gammac}{\gamma_\kappa}
\newcommand{\etac}{\eta_{\kappa}}
\newcommand{\Deltac}{\Delta_\kappa}

\newcommand{\growingmid}{\mathrel{}\middle|\mathrel{}}

\newcommand{\OPT}{\mathrm{OPT}}

\newcommand{\SW}{\mathrm{SW}}

\newcommand{\opton}{\mathrm{OPT}_\mathrm{on}}
\newcommand{\FPt}{\mathsf{FP}_t}
\newcommand{\SPt}{\mathsf{SP}_t}
\newcommand{\FPtj}{\mathsf{FP}_{t,j}}
\newcommand{\SPtj}{\mathsf{SP}_{t,j}}

\newenvironment{wrapper}[1]
{
	\smallskip
	\begin{center}
		\begin{minipage}{\linewidth}
			\begin{mdframed}[hidealllines=true, backgroundcolor=gray!20, leftmargin=0cm,innerleftmargin=0.35cm,innerrightmargin=0.35cm,innertopmargin=0.375cm,innerbottommargin=0.375cm,roundcorner=10pt]
				#1}
			{\end{mdframed}
		\end{minipage}
	\end{center}
	\smallskip
}

\newenvironment{wrapperEmpty}[1]
{
	\smallskip
	\begin{center}
		\begin{minipage}{\linewidth}
			\begin{mdframed}[hidealllines=false, backgroundcolor=white, leftmargin=0cm,innerleftmargin=0.35cm,innerrightmargin=0.35cm,innertopmargin=0.375cm,innerbottommargin=0.375cm,roundcorner=10pt]
				#1}
			{\end{mdframed}
		\end{minipage}
	\end{center}
	\smallskip
}

\begin{document}

\title{Approximating Optimum Online for Capacitated Resource Allocation\footnote{This work was done in part while the authors were visiting the Simons Institute for the Theory of Computing. Research supported in part in by Deutsche Forschungsgemeinschaft (DFG, German Research Foundation), Project No. 437739576, NSF Awards CCF2209520, CCF2312156, and a gift from CISCO. }}

\author{Alexander Braun\footnote{University of Bonn. alexander.braun@uni-bonn.de}  \quad Thomas Kesselheim\footnote{University of Bonn. thomas.kesselheim@uni-bonn.de} \quad Tristan Pollner\footnote{Stanford University. tpollner@stanford.edu} \quad Amin Saberi\footnote{Stanford University. saberi@stanford.edu}}

\date{\vspace{-1.3cm}}

\maketitle
\begin{abstract}
    We study online capacitated resource allocation, a natural generalization of online stochastic max-weight bipartite matching. This problem is motivated by ride-sharing and Internet advertising applications, where online arrivals may have the capacity to serve multiple offline users. 

Our main result is a polynomial-time online algorithm which is $(\nicefrac{1}{2} + \kappa)$-approximate to the optimal online algorithm for $\kappa = \approxconstant$. This can be contrasted to the (tight) $\nicefrac{1}{2}$-competitive algorithms to the optimum offline benchmark from the prophet inequality literature. Optimum online is a recently popular benchmark for online Bayesian problems which can use unbounded computation, but not ``prophetic'' knowledge of future inputs. 

Our algorithm (which also works for the case of stochastic rewards) rounds a generalized LP relaxation from the unit-capacity case via a two-proposal algorithm, as in previous works in the online matching literature. 
A key technical challenge in deriving our guarantee is bounding the positive correlation among users introduced when rounding our LP relaxation online. Unlike in the case of unit capacities, this positive correlation is unavoidable for guarantees beyond $\nicefrac{1}{2}$. Conceptually, our results show that the study of optimum online as a benchmark can reveal problem-specific insights that are irrelevant to competitive analysis.
\end{abstract}

\newpage
\setcounter{tocdepth}{3}
\tableofcontents

\newpage
\section{Introduction}

We study an \emph{online capacitated allocation problem}, in which users $\{1, 2, \ldots, n\}$ should be assigned to resources arriving online. Specifically, at each timestep $t \in \{1, 2, \ldots, T\}$, a new resource $t$ arrives and its capacity $c_t$ and the values $v_{i,t} \ge 0$ for every user $i \in [n]$ are sampled from 
a known distribution $\mathcal{D}_t$.
Upon the arrival of a resource, we observe its realized capacity and values, and must irrevocably decide which users to allocate to it. Our goal is to maximize social welfare, i.e., the sum of the values of assigned user-resource pairs. 

This problem naturally arises in a number of settings, for example in the context of ride-sharing: after a spike in demand (e.g. at the arrival of a flight, or at the end of a large concert), waiting passengers need to be assigned to cabs who become available online. Another example is online advertising, which initiated the vast literature on online Bayesian matching \cite{feldman2009online}, where ads should be assigned to search queries arriving online. Further examples are abundant: the assignment of orders to trucks by a shipping fulfillment center,  the procurement of goods for stores with limited inventories, etc. Our formulation goes beyond the intensely-studied setting where each online resource can be matched to at most one offline node (e.g. \cite{bahmani2010improved, haeupler2011online, manshadi2012online, jaillet2013online, ezra2020online, brubach2020online, huang2021online, tang2022fractional, huang2022power}). In many cases resources have capacity larger than one; multiple passengers can share a cab and multiple ads can be displayed under a search query. 

The literature has studied this problem from the ``prophet inequality'' perspective, designing algorithms which compare favorably to the \emph{optimum offline} algorithm which sees all realizations upfront. In particular, for online capacitated allocation, it is possible to obtain $\nicefrac{1}{2}$ of the optimum offline benchmark \cite{feldman2015combinatorial, dutting2020prophet}, and that is the best possible ~\cite{krengel1978semiamarts}. 

Still, comparing to the optimum offline algorithm as a benchmark might be too pessimistic in Bayesian settings. Its ``prophetic'' access to future realizations is unattainable for online algorithms (see \cite{papadimitriou2021online} for further discussion). Therefore, a recent line of work (also including, e.g., \cite{anari2019nearly, braverman2022max, dutting2023prophet}) has shifted attention towards the following question:  
\emph{how well can we approximate the optimal (computationally unbounded) online algorithm in polynomial time?}

In other words, how much must we lose when restricting to efficient algorithms instead of solving the optimal dynamic program? On the one hand, even for unit capacities it is PSPACE-hard to approximate the optimum online algorithm within some absolute constant $1-\epsilon$ \cite{papadimitriou2021online}.

Luckily, approximations strictly better than $\nicefrac{1}{2}$ exist for unit capacities: \cite{papadimitriou2021online} gave a $0.51$-approximate algorithm, later improved to 0.52 \cite{saberi2021greedy}, $1-\nicefrac{1}{\e} \approx 0.632$ \cite{braverman2022max}, and $0.652$ \cite{naor2023dependentroundingarxiv}. 
Motivated by this, we ask:

\begin{wrapperEmpty}
    \emph{Can we obtain a better than $\nicefrac{1}{2}$-approximate algorithm to the optimal (computationally unbounded) online algorithm beyond unit-capacity allocations?} 
\end{wrapperEmpty}

Our main result is to answer this question in the affirmative. In particular, we show that for online capacitated resource allocation problems, we can beat $\nicefrac{1}{2}$ by a constant. 

\begin{wrapper}
\begin{thm}\label{thm:beat_half_general_statement}
	For online capacitated allocation, there exists a polynomial time $\left(\nicefrac{1}{2} + \kappa\right)$-approximation to the social welfare of the optimal online algorithm, for a constant ${\kappa \ge \approxconstant}$. 
\end{thm}
\end{wrapper}

Interestingly, through the lens of prophet inequalities, the unit-capacity and the general capacity variant of the problem behave nearly identically. These variants (and more general ones) can all be handled by the same algorithmic template and techniques for the unit-capacity case directly carry over (for example, applying a $\nicefrac{1}{2}$-balanced online contention resolution scheme (OCRS) to each offline user). As we will discuss, studying capacitated resource allocation with the optimum online benchmark leads to technical challenges distinct from the unit-capacity case, and reveals differences that do not arise in competitive analysis. Our work hence gives evidence for the richness of studying optimum online as a benchmark. 

We also provide an extension to where allocations are probabilistically successful, motivated by our initial example of Internet advertising. As in the literature on online matching with stochastic rewards \cite{mehta2012online, huang2020onlinestochastic, goyal2023onlinestochastic}, after displaying ads under a search request, typically the advertiser is only charged if the ad is eventually clicked. This is typically modeled as happening with known probability called the click-through-rate. We hence update our setting so that after allocating at most $c_t$ users to the resource $t$, each user $i$ is \emph{successfully allocated} with known probability $q_{i,t}$. If an offline user $i$ is not successfully allocated, it remains available to be matched in future rounds; however, online arrivals do not get to adaptively pick new allocations. 

\begin{wrapper}
\begin{thm}\label{thm:beat_half_stochastic rewards}
	For online capacitated allocation with stochastic rewards, there exists a polynomial time $\left(\nicefrac{1}{2} + \kappa\right)$-approximation to the social welfare of the optimal online algorithm, for a constant ${\kappa \ge \approxconstant}$. 
\end{thm}
\end{wrapper}

Note \Cref{thm:beat_half_general_statement} is the special case of \Cref{thm:beat_half_stochastic rewards} in which all successes are deterministic, i.e. success probabilities $q_{i,t} = 1$ for every $i$, $t$.

\subsection{Our Techniques}
\label{sec:techniques}

Our algorithm rounds an LP relaxation online while introducing a controllable amount of positive correlation among offline users. For each online arrival $t$, we apply two rounds of pivotal sampling to the unallocated offline nodes, to guarantee never ``over-allocating'' $t$ beyond its remaining capacity. In each round, we only randomly allocate a subset of this sampled group to avoid large positive correlation between users. 

Throughout, in the main body of the paper, we focus on the special case where resource arrivals are ``Bernoulli'' (i.e., in step $t$, resource $t$ with known capacity $c_t$ and known values $v_{i,t} \ge 0$ arrives with probability $p_t$ and does not show up with probability $1-p_t$.)

\paragraph{LP relaxation.} In order to bound the social welfare achieved by the optimum online algorithm, we will use a linear program (LP) relaxation with variables $\{x_{i,t}\}$. The variables can be interpreted as the probabilities we would like an online algorithm to assign each user $i$ to resource $t$ (see \Cref{sec:preliminaries}). We require that at most $p_t \cdot c_t$ users are allocated to resource $t$ in expectation, and also make use of an ``online constraint'' which does not hold for offline algorithms, as in \cite{papadimitriou2021online, torrico2022dynamic}. In particular, for online algorithms, the arrival of resource $t$ and the event that user $i$ is unallocated at $t$ are independent. We account for stochastically successful allocations in this constraint and the LP's objective via the independence of success along edges from an online algorithm's allocation decisions. 

\paragraph{A two-proposal algorithmic approach.} Our algorithm rounds an optimal LP solution online such that (i) for every resource $t$, we do not allocate more users than its capacity, and (ii) every pair $(i,t)$ is successfully allocated with probability $(\nicefrac{1}{2} + \kappa) \cdot x_{i,t} \cdot q_{i,t}$ for a constant $\kappa := \approxconstant$. To achieve guarantees (i) and (ii), we use a two-proposal algorithm inspired by the algorithm used for matching by \cite{papadimitriou2021online}. For every resource $t$, we run up to two rounds; in each we propose to a subset of users whose size does not exceed the remaining capacity and a random subset of users accepts. 

While our LP relaxation and high-level framework are similar to \cite{papadimitriou2021online}, new ideas are needed for the specifics of the algorithm and (more importantly) its analysis. For example, when a new resource $t$ arrives, we would like to sample a subset of users such that $i$ is included with probability proportional to $x_{i,t}$. In the matching case, summing $x_{i,t}/p_t$ over all users $i$ never exceeds one, and hence, the vector $\left( x_{i,t}/p_t \right)_i$ naturally forms a probability distribution over users. This is no longer true in our more general capacitated allocation problem. Naïvely sampling users \emph{independently} with the given marginals has the issue that we might exceed the capacity of resource $t$. We instead rely on the technique of \emph{pivotal sampling} (also known as \emph{dependent rounding}) to ensure that the sampled set of users never exceeds the capacity of resource $t$ and that sampled users are negatively correlated. 
Via the pivotal sampling subroutine we get a \emph{first proposal set} of users. We note that in order to obtain this first proposal set, we apply the pivotal sampling in a history-agnostic way. That is, we may include previously successfully allocated users in the proposal set for resource $t$ at first. From the proposal set, we then randomly allocate each available user $i$ with some probability. 
This is essentially done according to a $(\nicefrac{1}{2} + \kappa )$-balanced online contention resolution scheme (OCRS). 

After this allocation process, there might remain a gap between the capacity of resource $t$ and the number of allocated users. We crucially exploit this gap by drawing a \emph{second proposal set} of users by another call of the pivotal sampling subroutine with reduced marginal probabilities. The reduction is precisely to ensure that the capacity of resource $t$ is not exceeded. Afterwards, we probabilistically assign a subset of these users, with a carefully chosen downsampling function. 

\paragraph{Analyzing the algorithm.}

For the analysis, we distinguish for each pair $(i,t)$ whether it is assigned with probability at least $(\nicefrac{1}{2} + \kappa) \cdot x_{i,t} \cdot q_{i,t}$ already from the first proposal, or requires the second proposal to reach this threshold. In the first case, the analysis proceeds in a straightforward way via the calculations originally from the OCRS literature (see e.g. \cite{ezra2020online}). 

For the remaining pairs $(i,t)$, bounding their contribution to the social welfare requires analyzing the second proposal, and doing so is our main technical contribution. 
For the second proposal along $(i,t)$ to contribute to social welfare, clearly user $i$ needs to be unallocated just before $t$ arrives. Furthermore, we are required to reduce the marginal probability that $i$ is sampled in the second proposal depending on the number of allocated users in $t$'s first proposal. This number in turn depends on the availability of \emph{other} users $j \neq i$. In particular, if a user $j$ is already assigned before the arrival of resource $t$, even when sampled as a first proposal, we cannot allocate the user. This increases the remaining capacity of resource $t$, which is beneficial for the marginal reduction required in our second proposal. 

Conversely, if we condition on other users $j \neq i$ being free before the arrival of resource $t$, it leads to a \emph{larger} decrease of the marginal probabilities in our second proposal. Still, this implies a decrease of the social welfare contribution of $(i,t)$. The relevant technical question, then, is if we condition on $i$ being free before $t$ arrives (necessary for $(i,t)$ to contribute to social welfare), how much can the conditional probabilities of users $j \neq i$ being free increase? Equivalently, how significantly can the availabilities of offline users be correlated? 

In the matching case this challenge was readily handled by showing negative correlation. While this is not possible for our problem (\Cref{subsec:poscorr}), we show that our algorithm obtains a good approximation if it can just avoid introducing ``large'' amounts of positive correlation. In the most technical part of our paper, we show our two-proposal algorithm achieves this by inductively tracking the availability of users over multiple rounds. In particular, we show the probability of both users being free at time $t$ is at most the product of the users' individual probabilities of being free, multiplied by $f\left(\sum_{t' < t}  x_{i,t'} \cdot q_{i,t'} \right)$, for $f(z) := 1 + z \cdot \left( \frac{(0.5+\kappa)^2}{1 - z \cdot (0.5 + \kappa)} \right)$. Interestingly, the point at which we evaluate the function $f$ does not depend on user $j$ at all.

\subsection{Capacitated Allocation Lacks Negative Correlation} \label{subsec:poscorr}

Even in the case where every success probability $q_{i,t}$ equals $1$, the potential positive correlation among offline users underlies the challenge for capacities exceeding 1. For example, a tempting naïve approach for general capacities is to directly reduce to the unit-capacity case: upon the arrival of a resource with capacity $c_t$, model this as $c_t$ resources with unit capacities, and simply run the algorithms from prior work. Unfortunately, this fails; a crucial assumption of the relevant literature is that arrivals are independent across different rounds, and introducing positive correlation across arrivals can be extremely problematic for existing algorithms. For example, consider the natural generalization of the algorithm by \cite{braverman2022max}: in round $t$, let users propose to the arriving resource and allocate the $c_t$ proposing users with the highest values. Here the positive correlation introduced can create severe problems. 

\begin{restatable}{obs}{bdmzeroapprox}
\label{example:bdm0approx}
For any $\epsilon > 0$, there exists an online capacitated allocation instance where the (generalized) algorithm of \cite{braverman2022max} is no more than $\epsilon$-approximate with respect to the welfare achieved by the optimal (computationally unbounded) online algorithm. 
\end{restatable}

The formal proof can be found in \Cref{app:informative-examples}. The approach of \cite{braverman2022max} is LP-based and one of the crucial steps is an upper bound on the probability that a subset of users is matched simultaneously. Intuitively speaking, their bound can be interpreted as a form of negative correlation among the offline users with respect to the LP variables.\footnote{It is possible to extend the algorithm of \cite{braverman2022max} to one with the same approximation ratio which furthermore has full negative correlation between offline nodes.} Unfortunately, simple examples show that in our case positive correlation is \emph{required} to go beyond an approximation ratio of $\nicefrac{1}{2}$.

\begin{restatable}{obs}{positivecorrelationrequired}
\label{observation:positivecorrelationrequired}
    Any algorithm for online capacitated allocation which has an approximation ratio better than $\nicefrac{1}{2}$ with respect to \textup{\eqref{LP}} must create positive correlation between the events of offline users being available. 
\end{restatable}

A formal version of the argument can be found in \Cref{app:informative-examples}. We note that the proof even rules out the ``negative correlation with respect to the LP'' showed by \cite{braverman2022max}, also used by follow-up work \cite{naor2023dependentroundingarxiv}. 
In contrast to the line of work by \cite{braverman2022max} and \cite{naor2023dependentroundingarxiv}, Papadimitriou \emph{et al.} \cite{papadimitriou2021online} gave a different algorithm for the unit-capacity case which operates in the mentioned ``two-proposals framework'' that has been successful for multiple problems in the online matching literature \cite{feldman2009online, manshadi2012online}. Critically, their analysis shows that almost all of the matches create negative correlation of offline nodes (in fact, satisfying the very strong property of \emph{negative association}). While our algorithm is inspired by the two-proposals framework, the example above demonstrates that there is no reasonable way to generalize this statement to the capacitated case while beating a $\nicefrac{1}{2}$-approximation. 

\subsection{Interlude on an Equivalent View: Online Combinatorial Auctions}
\label{sec:view_combinatorial_auctions}

Online capacitated resource allocation problems can also be interpreted in the context of \emph{online combinatorial auctions} --- a commonly studied setting in the prophet inequality literature, as in e.g. \cite{feldman2015combinatorial, dutting2020prophet, correa2023subadditive} and many others. Here, online arrivals correspond to \emph{buyers} and offline nodes are \emph{items}. Our capacities translate to the assumption that each buyer $t$ has a \emph{$c_t$-demand} valuation function, interpolating between unit-demand and fully additive valuations. 

In online capacitated allocation, we assume valuations are given upfront to the algorithm designer through a centralized planner. This view is (at first glance) less realistic for online combinatorial auctions --- here we would expect buyers to report their own valuations, and would need to consider incentives. Luckily, applying recent work of \cite{banihashem2024power}, we can argue that  our algorithm can be made dominant strategy incentive-compatibility (DISC) if we bound the demand size of buyers by a constant (a reasonable assumption for motivating applications). In particular, \Cref{thm:beat_half_general_statement} implies the following result which we formally prove in \Cref{app:beat_half_comb_auc}. 

\begin{wrapper}
\begin{restatable}{thm}{beathalfforcombauc}
\label{thm:beat_half_for_comb_auc}
	Say every buyer $t$ samples a $c_t$-demand valuation function, where $c_t$ is upper bounded by a constant. Then, for online combinatorial auctions, there exists a polynomial-time DSIC mechanism giving a $\left(\nicefrac{1}{2} + \kappa\right)$-approximation to the social welfare of the optimal online algorithm. 
\end{restatable}
\end{wrapper}

Note that our main \Cref{thm:beat_half_general_statement} does not require any upper bounds on the capacities $c_t$. In particular, the capacities $c_t$ can be as large as the number of offline users. The upper bound on $c_t$ in the combinatorial auction interpretation is only required such that the reduction from \cite{banihashem2024power} runs in polynomial time.

\subsection{Additional Related Work}

Online resource allocation problems have gained attention in the last decades due to a plethora of applications introduced by large marketplaces (see e.g. \cite{mehta2013online}). 

A particularly well-studied variety of such problems is \emph{online matching}. As initiated by \cite{karp1990optimal}, here we have a set of offline vertices and a set of vertices arriving online. Upon arrival, online nodes reveal a subset of offline nodes they could be matched to, and we can allocate at most one that is still available. \cite{karp1990optimal} give an online algorithm for this problem that achieves a $(1-1/e)$-approximation to the value of the best possible matching in hindsight. This guarantee was later extended to vertex-weighted instances, where offline vertices might have different values \cite{aggarwal2011online}. The case we consider where edges are only successful with known probability has also been studied in the literature, often going by \emph{online matching with stochastic rewards} \cite{mehta2012online, mehta2015onlinestochastic, huang2020onlinestochastic, goyal2023onlinestochastic, huang2023onlinestochastic}. When online nodes can adaptively attempt to ``rematch'' based on the successful status of edges, the problem is often called \emph{stochastic probing}, and it has been studied in both online \cite{bansal2012lp, adamczyk2015improved, borodin2020bipartite} and offline settings \cite{chen2009approximating, adamczyk2011improved, gamlath2019beating}. 

In the most general edge-weighted case, it is unfortunately impossible to obtain \emph{any} constant-factor approximation for adversarial arrivals; a recent line of work studies the case where we relax the requirement of decisions being irrevocable \cite{fahrbach2020edge, gao2021improved, blanc2021multiway}. But in settings where allocations cannot be easily reversed, the only other option is to move beyond the pessimistic assumption of fully adversarial arrivals. The most natural way to do so is to consider the intermediate model of stochastic arrivals, a reasonable assumption for settings with large amounts of historical data available. There is a long line of work designing matching algorithms in such settings, including edge-weighted problems (e.g. \cite{haeupler2011online, alaei2012online, brubach2016new, ezra2020online}) and vertex-weighted/unweighted problems (e.g. \cite{feldman2009online, manshadi2012online, jaillet2013online, huang2021online, huang2022power}). There is also very recent work studying \emph{correlated} arrivals in online stochastic matching \cite{aouad2023nonparametric}, showing guarantees of half against the offline benchmark when online nodes are independent across different types rather than arrival rounds.

Most of the literature on Bayesian online resource allocation problems focuses on competitive algorithms against the expected offline optimum, also called prophet inequalities. 
Originally introduced in the 70s and 80s by \cite{krengel1978semiamarts} and \cite{samuel1984comparison}, statements of this form gained renewed attention in the past decades due to connections with mechanism design \cite{hajiaghayi2007automated, chawla2010multi, kleinberg2019matroid}. In these mechanisms, a sequence of buyers arrives one-by-one and faces item prices, buying the most desirable feasible bundle. These mechanisms are incentive compatible and individually rational by design and lead to desirable approximation guarantees of the optimum achievable welfare. This explains the rise of literature in this area during the recent years \cite{feldman2015combinatorial, dutting2020prophet, dutting2020log, gravin2019prophet, correa2022pricing, braun2023simplepricing}.  
For more details, we refer to the survey by Lucier~\cite{lucier2017economic}.

Typical problems studied in the literature are weighted bipartite matching (a.k.a. unit-demand combinatorial auctions) as well as its generalizations towards more general scenarios, such as XOS or subadditive valuations in combinatorial auctions \cite{feldman2015combinatorial, dutting2020prophet, dutting2020log, correa2023subadditive}. In complementing work, also feasibility constraints such as (poly-)matroids \cite{duetting2015polymatroid, kleinberg2019matroid, chawla2020matroid}, knapsacks \cite{dutting2020prophet, jiang2022prophet} and beyond \cite{gobel2014online, rubinstein2016prophet, baek2019prophet} are considered. 

The paradigm of \emph{online contention resolution schemes} (OCRS) has been an influential technique for proving prophet inequalities. Here, we start with an LP relaxation of the offline allocation problem and run a rounding procedure online while observing realizations one-by-one. Introduced by \cite{feldman2016online}, this technique has been broadly applied, see e.g. \cite{lee2018optimal, ezra2020online, pollner2022ocrs, fu2022oblivious, avadhanula2023dynamicocrs, macrury2023ocrs}. The LP relaxation we use for our algorithm differs from standard OCRS settings as there are additional constraints in our LP which are \emph{only} valid for online algorithms. 

Online allocation has also been studied in the literature where offline nodes have capacities and can be allocated simultaneously in different rounds \cite{alaei2014bayesian, alaei2013gap}. For example, \cite{alaei2012online} study such a setting and derive competitive ratios against the offline benchmark which can be improved beyond $\nicefrac{1}{2}$ once there is a lower bound of at least 2 on the offline capacities. The literature has also considered the impact of reusability of offline nodes \cite{feng2019reusable, dickerson2021reusable, feng2022reusable}. 

\subsection{Paper Organization}

In \Cref{sec:preliminaries}, we formally state our problem and review some preliminaries. In \Cref{sec:algorithm}, we introduce our algorithm and argue that it is well-defined. Afterwards, in \Cref{sec:analysis}, we analyze the algorithm's approximation ratio, the main technical contribution of our work. We conclude in \Cref{sec:conclusion} with some future directions suggested by our work. 
\Cref{app:informative-examples} contains a discussion of informative examples and observations for our problem. In \Cref{app:deferredproofs}, we give proofs that are deferred from the main body. 

In the first part of the main body of our paper, we prove a simpler result for ease of exposition; the remaining sections and appendices include the details required to prove our result in full generality. Our algorithm as stated in \Cref{sec:algorithm} requires an exponential-time computation; in \Cref{app:sample_based_algo} we analyze the natural Monte Carlo variant and hence provide a truly polynomial-time algorithm. Our algorithm in \Cref{sec:algorithm} also focuses on the special case of Bernoulli arrivals; in \Cref{app:beyond_bernoulli} we show how to extend our techniques to online arrivals with values and capacities drawn from general distributions. Finally, for simpler notation, when analyzing our algorithm we consider the special case where every success probability $q_{i,t}$ is one; in \Cref{app:stochasticrewards} we discuss the necessary changes to prove the result for arbitrary probabilities. 

\section{Formal Problem Statement and Preliminaries} \label{sec:preliminaries}

In the following section, we will give a formal definition of a special case of our problem. For ease of exposition, in the first part of the main body of our paper we describe our algorithm and analysis for this special case, and list the additional details required to solve the general version only afterwards. We also will review some preliminaries including statements about our LP relaxation and the basics of pivotal sampling, an important ingredient for our algorithm. 

\paragraph{Problem definition.}
Recall that we defined the input to our problem as a set of $n$ users $I$ which are available offline. In addition, there is a set of resources $[T]$ which are revealed online in known order. In step $t$, resource $t$ arrives (also noted as active) independently with known probability $p_t$. 
In addition, value $v_{i,t} \ge 0$ is user $i$'s value for being served by resource $t$. Every user can be served by at most one resource; any resource can serve up to $c_t$ many users. We call $c_t$ the \emph{capacity} of resource $t$ and emphasize that $c_t$ can be resource-specific, i.e. we allow different resources to have different capacities. 
Upon the arrival of resource $t$, we observe the random realization if the resource is active, and can choose which users $I_t \subseteq I$ (if any) we would like to allocate to it, subject to the constraints that each user can be assigned to at most one resource and $|I_t| \le c_t$. If resource $t$ does not arrive, for convenience, we take $I_t = \emptyset$. 

Upon assigning $I_t$, each $i \in I_t$ is \emph{successfully allocated} with probability $q_{i,t}$ independently. We denote the successful set by $S(I_t)$. More generally, the set $S(J) \subseteq J$ denotes the set of successful allocations from some allocated set $J$.
Our objective is to maximize the \emph{expected social welfare}, defined as $\E \left[ \SW \right]  := \E \left[ \sum_t  \sum_{i \in S(I_t)} v_{i,t}  \right] $.

Our goal is to design a polynomial-time approximation algorithm for this problem. An algorithm is a \emph{$\zeta$-approximation} if for any instance of the problem, we have $ \E \left[ \SW \right] \geq \zeta \cdot \opton$, where $\opton $ is the expected welfare achieved by the optimal online algorithm. The optimal online algorithm has unlimited computational power and also knows all distributions upfront, but only observes realizations one at a time and needs to make an irrevocable decision before observing the next realization. Formally, we can define $\opton$ via a Bellman equation. To this end, let $\opton(t,J)$ denote the optimum gain achievable from resources $\{t, t+1, \ldots, T\}$ with users $J \subseteq I$ available. Then, recursively we have 
\begin{align*}
\opton(t,J) &:= (1 - p_t) \cdot \opton(t+1, J) \\&\quad \quad 
+ p_t \cdot \max_{J' \subseteq J, |J'| \le c_t}   \E \left[ \sum_{i \in S(J')} v_{i,t} + \opton(t+1, J \setminus S(J')) \right].
\end{align*}
We recall that even in the case of unit capacities with deterministically successful assignments, it is PSPACE-hard to approximate $\opton$ within a $(1-\epsilon)$ factor \cite{papadimitriou2021online}.

\paragraph{LP relaxation.}

We will use an LP relaxation of the optimum online algorithm which generalizes that for the unit-capacity and deterministic rewards case \cite{papadimitriou2021online, braverman2022max, torrico2022dynamic}. It has a variable $x_{i, t}$ for every pair of a user $i$ and a resource $t$.

\begin{align}
	\nonumber  \max \  &  \sum_{i,t} x_{i,t} \cdot q_{i,t} \cdot v_{i,t} && \tag{LP\textsubscript{on}} \label{LP} \\
	\text{s.t. }&\sum_{i} x_{i,t} \le p_t \cdot c_t && \text{for all } t \in [T] \label{eqn:userptktconstraint}\\
	&\  0 \le x_{i,t} \le p_t \cdot \left( 1 - \sum_{t' < t} x_{i, t'} \cdot q_{i,t'} \right) && \text{for all } i \in I, t \in [T] \label{eqn:PPSWConstraint}
 \end{align}

This LP indeed relaxes the optimal online algorithm: set $x_{i, t}$ to be the marginal probability that this algorithm attempts to allocate $i$ to $t$. 
Constraint~\eqref{eqn:userptktconstraint} holds as any algorithm can only allocate at most $c_t$ users to resource $t$ if it arrives. Constraint~\eqref{eqn:PPSWConstraint} only holds for online algorithms: the event of users being not yet successfully allocated at step $t$ and the event of resource $t$ arriving are independent. We note it implies the natural constraint that $\sum_t x_{i,t} \cdot q_{i,t} \le 1$.\footnote{Indeed, we can apply Constraint~\eqref{eqn:PPSWConstraint} to $(i,T)$ and observe $\sum_{t} x_{i,t } \cdot q_{i,t} \le x_{i,T} + \sum_{t' < T} x_{i,t'} \cdot q_{i,t'} \le x_{i,T} + 1 - \frac{x_{i,T}}{p_T} \le 1.$} 

\begin{restatable}{obs}{observationlprelaxopton}
\label{observation:LP_relax_OPT_on}
	The optimum objective value of \textup{\eqref{LP}} upper bounds the gain of optimum online, i.e., $\OPT\eqref{LP} \geq \opton$. 
\end{restatable}

For completeness the short formal proof is included in \Cref{app:LP_relax_OPT}.

\paragraph{Generalized problem definition.} In the above problem definition, we made the simplifying assumption that the resource arriving at time $t$ has a simple ``Bernoulli'' distribution determining if it is active or not. In the general model, in every round, a resource randomly realizes one of many possible pairs of valuation vectors to the users and capacities. Formally, in our general model, resource $t$ realizes one of $m$ possible capacities $c_{t,j}$ together with a vector of values $(v_{i,t,j})_i$, where each realization $j$ is sampled with probability $p_{t,j}$. We highlight that capacities and values during a single round $t$ can be arbitrarily correlated, although across different rounds we assume independence. In \Cref{app:beyond_bernoulli} we argue that our LP, algorithm, and analysis extend to such general settings as well. 

\subsection{Pivotal Sampling}

As a part of our online algorithm we invoke the randomized offline rounding framework of \emph{pivotal sampling} (also called \emph{Srinivasan rounding} and \emph{dependent rounding}) \cite{srinivasan2001distributions, gandhi2006dependent}. Imagine we are given marginals $x_1, \ldots, x_n$ with each $x_i \in [0,1]$ and $\sum_i x_i \le k$ for some positive integer $k$. We would like to randomly select at most $k$ indices from $\{1, 2, \ldots, n\}$ such that $i$ is selected with probability $x_i$. Pivotal sampling selects such a subset while also guaranteeing strong negative correlation properties between individual indices. It does so by sequentially choosing a pair of fractional marginals, and applying a randomized ``pivot'' operation that makes at least one integral.   We formally  state some of the properties of the algorithm below which suffice for our analysis.

\begin{thm}[as in \cite{srinivasan2001distributions}]
	The pivotal sampling algorithm with input $(x_i)_{i=1}^n$ where $\sum_i x_i \le k$ efficiently produces a random subset of $[n]$, denoted $\textup{\textsf{PS}}(x_1, \ldots, x_n)$, with the following properties:
	\begin{enumerate}[label=\text{(P\arabic*)}]
		\item For every $i \in [n]$, we have $\Pr[i \in \textup{\textsf{PS}}(x_1, \ldots, x_n)] = x_i$. \label{pivot_sampling_marginals}
		\item The number of elements in $\textup{\textsf{PS}}(x_1, \ldots, x_n)$ is always at most $k$. \label{pivot_sampling_bound_allocation_size}
		\item (Negative cylinder dependence) For any $I \subseteq [n]$, we have $$\Pr \left[ \bigwedge_{i \in I} i \in \textup{\textsf{PS}}(x_1, \ldots, x_n) \right] \le \prod_{i \in I} \Pr[i \in \textup{\textsf{PS}}(x_1, \ldots, x_n)] $$ and $$ \Pr \left[ \bigwedge_{i \in I} i \notin \textup{\textsf{PS}}(x_1, \ldots, x_n) \right] \le \prod_{i \in I} \Pr[i \notin \textup{\textsf{PS}}(x_1, \ldots, x_n)] \enspace. $$ \label{pivot_sampling_negative_dependence}
	\end{enumerate}
\end{thm}
\section{The Algorithm: A Two-Step Approach} \label{sec:algorithm}

We begin by a short description of our algorithm, before presenting the pseudocode in \Cref{allocationalgexact}. First we fix some useful definitions: we say user $i$ is ``allocated to $t$'' if it is one of the at most $ c_t$ users served by the resource, and ``successfully allocated to $t$'' if it is allocated to $t$ and $(i,t)$ is successful (recall this is with probability $q_{i,t}$).  We say user $i$ is ``free at $t$'' or ``available at $t$'' (or ``free''/``available'', if the context is clear) if just before the arrival of resource $t$, user $i$ has not yet been successfully allocated to any previous resource. 

Our algorithm uses an optimal solution $\{x_{i,t}\}$ to \eqref{LP} as input. After observing if resource $t$ arrives, if so, we sample a set of at most $c_t$ users $\FPt$ (denoting the \textbf{f}irst \textbf{p}roposal for $t$) using pivotal sampling, such that each user $i$ is selected with marginal probability $\nicefrac{x_{i,t}}{p_t}$. For every user $i \in \FPt$, if $i$ is still available, we toss a coin independently with probability $\alpha_{i,t} := \min \left(1, \frac{0.5 + \kappa }{1 - (0.5 + \kappa) \cdot \sum_{t' < t} x_{i,t'} \cdot q_{i,t'}} \right)$, and allocate user $i$ to resource $t$ if this coin toss is successful. 

After this procedure, we have a number $A_t$ of users allocated to resource $t$, where $A_t$ is a random variable which can take values in $\{0,\dots, c_t\}$. In order to make use of the remaining space in the demand size of resource $t$, we allow $t$ to make a second proposal. Again via the pivotal sampling subroutine, this time with a reduced marginal probability of $(1 - \frac{A_t}{c_t}) \cdot x_{i,t}/p_t $ for every user $i$, we sample a set of users $\SPt$, denoting the \textbf{s}econd \textbf{p}roposal with size at most $c_t - A_t$. Among these users, we consider only those $i$ for which $\alpha_{i,t} = 1$, $i$ was free at $t$, and $i$ was not yet allocated to $t$. For each such user $i$, we allocate to $t$ with probability $\beta_{i,t}$. The factor $\beta_{i,t}$ is chosen in a way to ensure that $\Pr[i \text{ allocated to } t] = (0.5+\kappa) \cdot x_{i,t}$, i.e., such that we don't \emph{overmatch} any $(i,t)$.

\begin{algorithm}[H]
	\caption{}
	\label{allocationalgexact}
	\begin{algorithmic}[1]

		\State $\kappa \gets \approxconstant$ \label{line:sample_first}
  
		\State Solve \eqref{LP} for $\{x_{i,t} \}$  \label{line:solveLP}
  
		\For{each time $t$, if $t$ arrives} \Comment{w.p. $p_t$}
		
		\State Define users $\FPt := \textsf{PS}( ( x_{i,t}/p_t )_{i \in I})$  \Comment{at most $c_t$ users get first proposal}  \label{line:firstproposal}
		\For{each user $i \in \FPt$}
		\If{$i$ is available}
		\State Allocate $i$ to $t$ with probability $\alpha_{i,t} := \min \left(1, \frac{0.5 + \kappa }{1 - (0.5 + \kappa) \cdot \sum_{t' < t} x_{i,t'} \cdot q_{i,t'}} \right)$  \label{line:sample_alphait} \label{line:firstmatch}
		\EndIf 
		
		\EndFor
		
		\State Let $A_t \gets \text{number of users allocated to } t \text{ thus far}$ \label{line:sample_defAt}

		\State Define users $\SPt:= \textsf{PS}( ( (1 - \frac{A_t}{c_t}) \cdot x_{i,t}/p_t )_{i \in I})$ \label{line:secondproposal} \Comment{$\le c_t - A_t$ users get second proposal} 
		\For{each user $i \in \SPt$ with $\alpha_{i,t} = 1$}
		\If{$i$ is available and currently unallocated} \label{line:ifiavail}
        \State Compute $\rho_{i,t} := \E [ \mathbbm{1}[i \text{ available and unallocated after \Cref{line:sample_defAt}}] \cdot (1 - \frac{A_t}{c_t}) \mid t \text{ arrived}] $ \label{line:computerhoexact}
        \State $\beta_{i,t} \gets \min \Big(1, \left( (0.5 + \kappa) \cdot \sum_{t' < t} x_{i,t'} \cdot q_{i,t'} - (0.5 - \kappa) \right) \cdot \frac{1}{\rho_{i,t}} \Big).$
		\State Allocate $i$ to $t$ with prob. ${\beta_{i,t}}$ \label{line:secondmatch}
		\EndIf
		\EndFor
		
		\EndFor
	\end{algorithmic}
\end{algorithm}	

Concerning the definition of $\rho_{i,t}$, we note that the expectation is over the randomness in the arrivals and algorithm up to when it reaches \Cref{line:sample_defAt} for arrival $t$ in \Cref{allocationalgexact} (in particular, we consider ``re-running'' the algorithm as defined thus far on a fresh instance). The indicator $\mathbbm{1}[i \text{ available and unallocated after \Cref{line:sample_defAt}}]$ refers to the event that $i$ was not successfully allocated to some $t' < t$ and is also not allocated yet to $t$ (it could be the case that $i$ was allocated to some $t' < t$, and this was unsuccessful). This indicator is potentially correlated with the number of allocated users $A_t$. 

The $\min(1,\cdot)$ in the definition of $\beta_{i,t}$ is for convenience only; in particular, it is thus easy to see that the algorithm is well-defined. As a crux of our analysis, we will show that using $\kappa = \approxconstant$ ensures that the $\min(1,\cdot)$ in the definition of $\beta_{i,t}$ is actually redundant. 

In the remainder of this section, we will argue that \Cref{allocationalgexact} is well-defined and guarantees to respect the capacity constraints of online resources.

\begin{obs}
    \Cref{allocationalgexact} is well-defined. 
\end{obs}

\begin{proof}
    Note first that in \Cref{line:firstproposal}, our call to the pivotal sampling algorithm $\textsf{PS}(\cdot)$ is well-defined as each marginal $\nicefrac{x_{i,t}}{p_t}$ is in $[0,1]$ by \ref{LP} Constraint \eqref{eqn:PPSWConstraint}. Each $\alpha_{i,t}$ as defined in \Cref{line:sample_alphait} is clearly a probability by construction. Our second call to $\textsf{PS}(\cdot)$ is similarly well-defined. Note that $\beta_{i,t}$ is always a probability --- if $\alpha_{i,t} = 1$, it implies that $(0.5 + \kappa) \cdot \sum_{t' < t} x_{i,t'} \cdot q_{i,t'} \ge (0.5 - \kappa)$ by definition. This in turn shows that $\beta_{i,t}$ is always in the interval $[0,1]$. 

    Finally, note that user $i$ is allocated only if available, and hence never successfully allocated to two different resources (or to the same resource twice). 
\end{proof}

We also have that our algorithm respects capacity constraints for each online arrival.

\begin{obs} \label{obs:respectcapacity}
    The number of users allocated to resource $t$ by \Cref{allocationalgexact} is always at most $c_t$. 
\end{obs}

\begin{proof}
    By Property~\ref{pivot_sampling_bound_allocation_size} of pivotal sampling,  the size of $\textsf{FP}_t$ is never larger than $c_t$ as $\sum_i \frac{x_{i,t}}{p_t} \leq c_t$ by Constraint~\eqref{eqn:userptktconstraint}. In addition, as we scale the marginals down for the second proposal set $\textsf{SP}_t$, we are guaranteed that resource $t$ is only allocated at most $c_t - A_t$ many users during the second proposal. 
\end{proof}

We also note that every line except \Cref{line:computerhoexact} can be implemented in polynomial time. Indeed, note \Cref{line:solveLP} can be run efficiently as \eqref{LP} has polynomial size, and that our calls to pivotal sampling can be implemented efficiently \cite{srinivasan2001distributions}. 

\Cref{line:computerhoexact} requires exponential time as written, and for ease of presentation, in the next section we analyze the above exponential time algorithm. In \Cref{app:sample_based_algo} we show that we can replace this computation with a sample average and appeal to concentration bounds, while only losing an arbitrarily small $\epsilon$ in the approximation ratio. The main point of care is to argue that $\rho_{i,t}$ is bounded away from 0 so that we can get a close  multiplicative approximation. 

\section[Analysis: Beating a 1/2-Approximation]{Analysis: Beating a $\nicefrac{1}{2}$-Approximation} \label{sec:analysis}
 
Our main result is as follows. 

\begin{thm}\label{theorem:main_theorem}
    For $\kappa= \approxconstant$, the social welfare achieved by \Cref{allocationalgexact} satisfies $$  \E \left[ \SW  \right]  \geq (0.5 + \kappa) \cdot  \opton \enspace. $$
\end{thm}

This section is dedicated to the proof of our main result. As mentioned before, we analyze the algorithm which has access to the expectation $\rho_{i,t}$ exactly. Note that this requires exponential time; however, in \Cref{app:sample_based_algo} we show that our sampling-based estimation only results in an additional loss of $\epsilon$ in the approximation. To prove this, we will rely on a consequence of our analysis, namely that the quantity $\rho_{i,t}$ is always bounded away from zero by some constant. Using this, we can apply standard Chernoff-Hoeffding concentration bounds to get reasonably close to the exact $\rho_{i,t}$ within small multiplicative error. 

To simplify the exposition, we will additionally assume that every allocation is successful, i.e., each success probability $q_{i,t}$ equals 1. In \Cref{app:stochasticrewards} we outline the necessary steps to generalize our analysis to the case of arbitrary success probabilities $\{q_{i,t}\}_{i,t}$. 

\paragraph{Outline.} Before diving into details we outline the ingredients in our proof of Theorem~\ref{theorem:main_theorem}. Firstly we note that by \Cref{obs:respectcapacity}, the size of $I_t$ (the set of users allocated to $t$) is always at most $c_t$, so $$\E \left[ \SW \right]  = \E \left[ \sum_t \max_{S \subseteq I_t, |S| \le c_t} \left( \sum_{i \in S} v_{i,t} \right) \right] = \sum_{i,t} v_{i,t} \cdot \Pr[i \text{ allocated to } t].$$

We will note that bounding the term $\Pr[i \text{ allocated to } t]$ naturally brings us into one of two cases. If $(i,t)$ is such that $\alpha_{i,t} < 1$, the allocation of $i$ to $t$ can only happen in \Cref{line:firstmatch} of our algorithm, and consequently it is straightforward to bound the resulting welfare (which we do in \Cref{lemma:approx_early_edges}). We then turn our perspective towards pairs $(i,t)$ with a subsampling probability $\alpha_{i,t} = 1$; for these, the analysis requires much more care. Again, we start by considering the contribution of allocating via a first proposal in Lemma~\ref{lemma:approx_late_edges} (i). Here the first proposal alone is not sufficient, and we are required to compensate for this via a suitable bound on the allocation probability via a second proposal. We do so by proving Lemma~\ref{lemma:approx_late_edges} (ii) which gives a sufficient lower bound of the contribution via a second proposal. This is the main technical contribution and will use lemmas analyzing the evolution of the correlation between offline users in \Cref{subsection:correlation_bound}. 

\paragraph{Notation.} For convenience, we let $y_{i,t} := \sum_{t' < t} x_{i,t'}.$ Note that $\alpha_{i,t} < 1$ exactly when $y_{i,t} < (0.5-\kappa) \cdot (0.5 + \kappa)^{-1}$. We hence define $\tau := (0.5-\kappa) \cdot (0.5 + \kappa)^{-1}$ as this threshold for $y_{i,t}$ after which the subsampling probability $\alpha_{i,t}$ becomes one. If for resource $t$ and user $i$ we have $y_{i,t} \leq \tau$, then we call the pair $(i,t)$ \emph{early}. Otherwise, we call the pair $(i,t)$ \emph{late}. In addition, we define $\mathcal{A}_1$ as the set of all pairs $(i,t)$ such that user $i$ was allocated to resource $t$ in \Cref{line:firstmatch}, and $\mathcal{A}_2$ as the set of all pairs $(i,t)$ such that $i$ was allocated to $t$ in \Cref{line:secondmatch}. 

\medskip

As $i$ is not allocated more than once in our algorithm, we quickly observe the following claim. 

\begin{obs}\label{observation:split_value}
    For any resource $t$, we have $$ \E \left[ \max_{S \subseteq I_t, |S| \le c_t} \left( \sum_{i \in S} v_{i,t} \right) \right]  = \sum_{i \in I} v_{i,t} \cdot ( \Pr[(i,t) \in \mathcal{A}_1] + \Pr[ (i,t) \in \mathcal{A}_2]) . $$
\end{obs}

To analyze the probabilities $\Pr[(i,t) \in \mathcal{A}_1]$ and $\Pr[(i,t) \in \mathcal{A}_2]$, we consider two separate cases based on whether $(i,t)$ is early (\Cref{subsection:analysis_early}) or late (\Cref{subsection:analysis_late}). 

\subsection{Analysis for Early Pairs} \label{subsection:analysis_early}

It will be crucial to bound the probability of a user $i$ being free at time $t$. We denote the event that user $i$ is \emph{free} or \emph{available} (i.e., not allocated) at the arrival of resource $t$ by $F_{i,t}$. The following observation gives an expression of the probability with respect to the LP variables. It is crucial to note that if a pair $(i,t)$ is early, so is every pair $(i,t')$ with $t' < t$.

\begin{obs}\label{observation:early_edges}
	For early pairs $(i,t)$, we have $\Pr[F_{i,t}] = 1 - (0.5 + \kappa) \cdot y_{i,t} $.
\end{obs}
\begin{proof}
	We proceed via induction on $t$. Before the arrival of the first resource, the claim is trivially true, as all users are available with probability one. Afterwards, note that 
	\begin{align}
		\Pr[(i,t) \in \mathcal{A}_1] &= p_t \cdot \Pr[i \in \FPt] \cdot \Pr[F_{i,t}] \cdot \alpha_{i,t} \label{eqn:itinA1early}
	\end{align} as $t$'s arrival, $i$ being included in $\FPt$, $F_{i,t}$ and the algorithm's $\text{Ber}(\alpha_{i,t})$ coin flip are mutually independent events. If $(i,t)$ is early, then $\alpha_{i,t} = \frac{0.5 + \kappa }{1 - (0.5 + \kappa) \cdot y_{i,t}}$, so we have $$\Pr[(i,t) \in \mathcal{A}_1] = p_t \cdot \frac{x_{i,t}}{p_t} \cdot (1 - (0.5 + \kappa) \cdot y_{i,t}) \cdot \frac{0.5 + \kappa }{1 - (0.5 + \kappa) \cdot y_{i,t}} = (0.5 + \kappa) \cdot x_{i,t},$$ where we also use the induction hypothesis for the probability of the user being free at the arrival of resource $t$. For early $(i,t)$, we also clearly have $\Pr[(i,t) \in \mathcal{A}_2] = 0$, so 
 \begin{align*}
 \Pr[F_{i,t+1}] &= \Pr[F_{i,t}] - \Pr[(i,t) \in \mathcal{A}_1] = 1 - (0.5 + \kappa) \cdot y_{i, t+1}. \qedhere
 \end{align*}
\end{proof}

As a consequence we can bound the contribution of an early pair $(i,t)$ to $\mathcal{A}_1$ and $\mathcal{A}_2$, as follows. 

\begin{restatable}{obs}{lemmaapproxearlyedges} \label{lemma:approx_early_edges}
    For early pairs $(i,t)$, $\Pr[(i,t) \in \mathcal{A}_1] = (0.5 + \kappa) \cdot x_{i,t} $ and ${\Pr[(i,t) \in \mathcal{A}_2] = 0}$.
\end{restatable}

Thus for early pairs $(i,t)$, our algorithm achieves the desired allocation probability. 

\subsection{Analysis for Late Pairs implies \Cref{theorem:main_theorem}} \label{subsection:analysis_late}

For late pairs, we show the following lemma which will be sufficient to prove our main \Cref{theorem:main_theorem}. 

\begin{lemma}\label{lemma:approx_late_edges}
    For late pairs $(i,t)$, the following two statements hold: 
    \begin{itemize}
        \item[\textnormal{(i)}] $\Pr[(i,t) \in  \mathcal{A}_1] = (1 - (0.5 + \kappa) \cdot y_{i,t} ) \cdot x_{i,t} $, and \label{lemma:approxlateedges1}
        \item[\textnormal{(ii)}] $\Pr[(i,t) \in  \mathcal{A}_2] = ( (0.5 + \kappa) \cdot y_{i,t} - 0.5 + \kappa) \cdot x_{i,t} $.
    \end{itemize}
\end{lemma}

We note that this immediately implies our main result.

\begin{proof}[Proof of Theorem~\ref{theorem:main_theorem}.]
    We have $\Pr[(i,t) \in  \mathcal{A}_1] + \Pr[(i,t) \in  \mathcal{A}_2] = (0.5 + \kappa) \cdot x_{i,t}$ for any pair $(i,t)$ by \Cref{lemma:approx_early_edges} and \Cref{lemma:approx_late_edges}. Hence, using the decomposition in \Cref{observation:split_value}, we have 
    \begin{align*}
         \E \left[ \sum_t v_t (I_t) \right] & = \sum_t \sum_{i \in I} v_{i,t} \cdot ( \Pr[(i,t) \in \mathcal{A}_1] + \Pr[ (i,t) \in \mathcal{A}_2])  \\ &=   \sum_t \sum_{i \in I} v_{i,t} \cdot (0.5 + \kappa) \cdot x_{i,t} \\ & = (0.5+\kappa) \cdot \OPT \eqref{LP}   \geq (0.5+\kappa) \cdot \opton . \qedhere
    \end{align*}
\end{proof}

Thus, it remains to prove \Cref{lemma:approx_late_edges}. Our analysis here requires significantly more care as it must bound the gain from the second proposal. As the second proposal's marginal probabilities are dependent on which offline users were allocated in the first proposal, a complete analysis must consider the correlation introduced. 

\subsubsection{Proof of \Cref{lemma:approx_late_edges} (i)} 

As for early pairs, the remainder of our proof will proceed by induction on $t$. Thus, for every late pair $(i,t')$ with $t' < t$, by the inductive hypothesis we have $\Pr[(i,t') \in \mathcal{A}_1] + \Pr[(i,t') \in \mathcal{A}_2] = (0.5 + \kappa) \cdot x_{i,t'}$. Recall also that for every early pair $(i,t')$ we know from \Cref{lemma:approx_early_edges} that $\Pr[(i,t') \in \mathcal{A}_1] + \Pr[(i,t') \in \mathcal{A}_2] = (0.5 + \kappa) \cdot x_{i,t'}$. Thus, we may assume that for the late pair $(i,t)$ being considered we have \begin{align} \Pr[F_{i,t}] = 1 - (0.5 + \kappa) \cdot y_{i,t} . \label{eqn:Fit} \end{align} With this, bounding the probability of allocation along a first proposal is very straightforward. 

\begin{proof}[Proof of \Cref{lemma:approx_late_edges} (i).]
    Note that
    \begin{align*}
	\Pr[(i,t) \in \mathcal{A}_1] &= p_t \cdot \Pr[F_{i,t}] \cdot \Pr[i \in \textsf{FP}_t] \cdot \alpha_{i,t} && \text{(Equation \eqref{eqn:itinA1early})} \\
	&= p_t \cdot \left( 1 - (0.5 + \kappa) \cdot y_{i,t} \right) \cdot \frac{x_{i,t}}{p_t} \cdot 1 && \text{(\Cref{eqn:Fit})}\\
	&= (1 - (0.5 + \kappa) \cdot y_{i,t}) \cdot x_{i,t}. && \qedhere 
\end{align*}
\end{proof}

This completes the proof of Lemma~\ref{lemma:approx_late_edges} (i), and the remainder of this section is dedicated to the proof of Lemma~\ref{lemma:approx_late_edges} (ii). 

\subsubsection{Proof of \Cref{lemma:approx_late_edges} (ii)}
\label{sec:approx_late_edges}

We begin by bounding $\Pr[(i,t) \in  \mathcal{A}_2]$ for late pairs $(i,t)$, in the natural way which depends on the number of allocated users during the first proposal in \Cref{line:firstmatch}. (Recall that this is because for second proposals, we reduce the marginal probabilities for pivotal sampling algorithm by a factor of $1 - \nicefrac{A_t}{c_t}$). Note that for $(i,t)$ to be matched as a second proposal we need all of the following to happen: (i) $t$ should arrive, (ii) $i$ must be available and unallocated after \Cref{line:sample_defAt}, and  included as a second proposal, and (iii) the potential match $(i,t)$ should survive the final downsampling by $\beta_{i,t}$. This lets us observe

\begin{align}
	\Pr[(i,t) &\in \mathcal{A}_2] \nonumber\\
    &= p_t \cdot 
	\Pr[i \text{ available and unallocated after \Cref{line:sample_defAt}} \wedge i \in \textsf{SP}_t \mid t \text{ arrived}] \cdot \beta_{i,t}  \nonumber 
 \\
 &= p_t \cdot \E \left[ \mathbbm{1}[i \text{ available and unallocated after \Cref{line:sample_defAt}}] \cdot \left( 1 - \frac{A_t}{c_t} \right)\cdot \frac{x_{i,t}}{p_t}  \bigm| t \text{ arrived} \right] \cdot \beta_{i,t} \nonumber \\
 &= x_{i,t} \cdot \rho_{i,t} \cdot \beta_{i,t}. \label{eqn:rhotimesbeta}
\end{align}

For the second equality, we relied on Property \text{\ref{pivot_sampling_marginals}} of pivotal sampling, which guarantees that individual elements are sampled with exactly their marginal probability. Note that this marginal probability is random, and potentially correlated with $ \mathbbm{1}[i \text{ available and unallocated after \Cref{line:sample_defAt}}]$.

Recall that $\beta_{i,t} := \min \Big( 1 , \left( (0.5 + \kappa) \cdot y_{i,t} - (0.5 - \kappa) \right) \cdot \frac{1}{\rho_{i,t}} \Big)$. If the $\min(1, \cdot)$ here is redundant, we are immediately done; this is concretized in the following observation. 

\begin{obs}\label{obs:doneifbetaitatmost1}
    If $\rho_{i,t} \ge (0.5 + \kappa) y_{i,t} - (0.5 - \kappa)$, then $$\Pr[(i,t) \in \mathcal{A}_2] = x_{i,t} \cdot  \left( (0.5 + \kappa) \cdot y_{i,t} - (0.5 - \kappa) \right). $$
\end{obs}

Thus it suffices to show that the hypothesis of this observation holds. In other words, for the remainder of the proof, the only thing we need to show is the following proposition.

\begin{prop}
    For any late pair $(i,t)$, we have $\rho_{i,t} \ge (0.5 + \kappa) y_{i,t} - (0.5 - \kappa)$.
\end{prop}

As a first step, we start with the following lower bound on $\rho_{i,t}$.

\begin{lemma}\label{lem:rhoitlowerbound}
For late pairs $(i,t)$, $$\rho_{i,t} \ge (1 - (0.5 + \kappa ) \cdot y_{i,t} ) \cdot \left(\tau -  \frac{\E[A_t \mid t \textup{ arrived}, F_{i,t}]}{c_t} \right).$$
\end{lemma}
\begin{proof}[Proof of \Cref{lem:rhoitlowerbound}.]
    Note first that we can expand 
    \begin{align*}
    \rho_{i,t} &=  \E \left[ \mathbbm{1}[i \text{ available and unallocated after \Cref{line:sample_defAt}}] \cdot \left( 1 - \frac{A_t}{c_t} \right)   \mid t \text{ arrived} \right] \\
    &= \Pr[F_{i,t} \mid t \text{ arrived}] \cdot \E \left[ \mathbbm{1}[i \text{ not allocated in \Cref{line:firstmatch}}] \cdot \left( 1 - \frac{A_t}{c_t} \right)   \mid t \text{ arrived}, F_{i,t} \right] \\
    &= \Pr[F_{i,t} ] \cdot \E \left[ \mathbbm{1}[i \text{ not allocated in \Cref{line:firstmatch}}] \cdot \left( 1 - \frac{A_t}{c_t} \right)   \mid t \text{ arrived}, F_{i,t} \right].
\end{align*}
Note that as the pair $(i,t)$ is late, we have $\alpha_{i,t} = 1$. Hence, conditioned on being free and the arrival of resource $t$, user $i$ is not allocated in \Cref{line:firstmatch} if and only if it is not contained in the set $\FPt$. This allows us to bound 
\begin{align*}
	\E \Bigg[ \mathbbm{1}[i \text{ not allocated in \Cref{line:firstmatch}}] &\cdot \left( 1 - \frac{A_t}{c_t} \right) \mid t \text{ arrived}, F_{i,t}  \Bigg] \\
	&= \E \left[ \mathbbm{1}[i \notin \textsf{FP}_t] \cdot \left( 1 - \frac{A_t}{c_t} \right) \mid t \text{ arrived}, F_{i,t} \right] \\
	&= \left( 1 - \frac{x_{i,t}}{p_t} \right) \cdot \E \left[ \left( 1 - \frac{A_t}{c_t} \right) \mid t \text{ arrived}, F_{i,t}, i \notin \textsf{FP}_t\right] \\
	&\ge \tau \cdot \E \left[ \left( 1 - \frac{A_t}{c_t} \right) \mid t \text{ arrived}, F_{i,t}, i \notin \textsf{FP}_t \right]. 
\end{align*}
To reason about the resulting expectation, we first apply the following bounding to remove the conditioning on $i \notin \textsf{FP}_t$:
\begin{align*}
    \E[A_t \mid t \text{ arrived}, F_{i,t}, i \notin \FPt] & = \frac{\E[A_t \cdot \mathds{1}_{i \notin \FPt} \mid t \text{ arrived}, F_{i,t}]}{\Pr[i \notin \FPt \mid t \text{ arrived}, F_{i,t}]} \\ & \le \frac{\E[A_t  \mid t \text{ arrived}, F_{i,t}]}{\Pr[i \notin \FPt \mid t \text{ arrived}, F_{i,t}]} \enspace .
\end{align*}
In addition, note that $\Pr[i \notin \FPt \mid t \text{ arrived}, F_{i,t}] = 1 - \frac{x_{i,t}}{p_t} \geq y_{i,t} \geq \tau$ as pair $(i,t)$ is late. Thus we get
\begin{align*}
	\E[A_t \mid t \text{ arrived}, F_{i,t}, i \notin \textsf{FP}_t] \le \frac{1}{\tau} \cdot \E[A_t \mid t \text{ arrived}, F_{i,t}] .
\end{align*}

By substitution and using \Cref{eqn:Fit}, we directly conclude 
\begin{align}
\rho_{i,t} & \ge \Pr[F_{i,t}] \cdot \tau \cdot \left( 1 - \frac{\E[A_t \mid t \text{ arrived}, F_{i,t}]}{c_t} \cdot \tau^{-1} \right) \label{eqn:rhoitlowerbound} \\
&= (1 - (0.5 + \kappa ) \cdot y_{i,t} ) \cdot \left(\tau -  \frac{\E[A_t \mid t \text{ arrived}, F_{i,t}]}{c_t} \right). && \nonumber \text{(via \Cref{eqn:Fit})} 
\end{align}
as claimed.
\end{proof}

In order to exploit the bound obtained in Lemma~\ref{lem:rhoitlowerbound}, we need to control $\E[A_t \mid t \text{ arrived}, F_{i,t}]$. In particular, our goal is to show that $\E[A_t \mid t \text{ arrived}, F_{i,t}]$ is bounded away from $c_t$ by a multiplicative constant smaller than $1$. If there was no conditioning on $F_{i,t}$, it is easy to check that $$\E[A_t \mid t \text{ arrived}] = \sum_j \Pr[F_{j,t}] \cdot \Pr[j \in \textsf{FP}_t] \cdot \alpha_{j,t} \le (0.5 + \kappa) \cdot c_t.$$ The conditioning could however lead us into trouble in the following way: When facing the conditioning, we end up with the expression $$ \E[A_t \mid t \text{ arrived}, F_{i,t}] = \sum_j  \Pr[F_{j,t} \mid F_{i,t}] \cdot \Pr[j \in \textsf{FP}_t] \cdot \alpha_{j,t}  . $$ If $F_{i,t}$ implies $F_{j,t}$ for every $j \neq i$, and $\alpha_{j,t} \approx 1$ for every $j \neq i$, then $${\E[A_t \mid t \text{ arrived}, F_{i,t}]} \approx \sum_i 1 \cdot \frac{x_{i,t}}{p_t} \cdot 1$$ where the right-hand side could equal $c_t$. This, in particular, would make the second proposal in our algorithm completely useless as we would reduce the marginal probabilities for the pivotal sampling in \Cref{line:secondproposal} to (almost) zero. The most crucial part of our analysis is to demonstrate that this cannot happen, by bounding the possible positive correlation introduced between offline users. 

\begin{restatable}{lem}{corboundcorrelationunrestricted}\label{corollary:bound_correlation_unrestricted} For any distinct users $i$ and $j$, and $\Deltac :=  \left( 1 + \frac{(0.5+\kappa)^2}{0.5-\kappa} \right) \cdot \left( \frac{0.5 + \kappa}{0.5 - \kappa} \right)^2 $, for any $t$ we have $$\Pr[F_{i,t} \wedge F_{j,t}] \le \Deltac  \cdot \Pr[F_{i,t}] \cdot \Pr[F_{j,t}].$$
\end{restatable}

The proof of \Cref{corollary:bound_correlation_unrestricted} is deferred to \Cref{subsection:correlation_bound}; in the remainder of this section we demonstrate why it implies our bound on the approximation ratio. We note that for $\kappa = \approxconstant$ (the value we choose in \Cref{allocationalgexact}), we have $\Deltac \approx 1.68$. As a concrete example, note that if $(i,t)$ and $(j,t)$ are both late with $\Pr[F_{i,t}] \approx \Pr[F_{j,t}] \approx \nicefrac{1}{2}$, this bound quantifies that we avoid perfect positive correlation between $F_{i,t}$ and $F_{j,t}$. 

Having \Cref{corollary:bound_correlation_unrestricted}, we can prove the bound on $\E[A_t \mid t \text{ arrived}, F_{i,t}]$ which we state formally in  \Cref{lemma:expected_At_with_condition} via 
\begin{align}
	\E[A_t \mid t \text{ arrived}, F_{i,t}] &= \sum_j  \Pr[F_{j,t} \mid F_{i,t}] \cdot \Pr[j \in \textsf{FP}_t] \cdot \alpha_{j,t} \nonumber \\
 &= \frac{x_{i,t}}{p_t} + \sum_{j \neq i} \frac{\Pr[F_{i,t} \wedge F_{j,t}]}{\Pr[F_{i,t}]} \cdot \frac{x_{j,t}}{p_t} \cdot \alpha_{j,t} \label{eqn:AtconditionedtarrivedFit} \\
	&\le \frac{x_{i,t}}{p_t} + \sum_{j \neq i} \Deltac \cdot \Pr[F_{j,t}] \cdot \frac{x_{j,t}}{p_t} \cdot \alpha_{j,t} \nonumber \\
	&\le \frac{x_{i,t}}{p_t} + \Deltac \cdot (0.5 + \kappa) \cdot c_t  \nonumber .
 \end{align}
The last inequality uses the fact that $\Pr[F_{j,t}] \cdot \alpha_{j,t} \le 0.5 +\kappa $ and upper bounds $ \sum_{j \neq i} \frac{x_{j,t}}{p_t}$ by $c_t$. By the online constraint \eqref{eqn:PPSWConstraint} and the property that $y_{i,t} > \tau$ for late pairs $(i,t)$, we have that $\frac{x_{i,t}}{p_t} \leq 1 - \tau$. Hence, we can conclude that
 \begin{align}\label{eqn:bound_conditional_At}
	\E[A_t \mid t \text{ arrived}, F_{i,t}] & \le 1 - \tau  + \Deltac \cdot (0.5 + \kappa ) \cdot c_t \\
	&\le \left( 1 - \tau  + \Deltac \cdot (0.5 + \kappa)  \right) \cdot c_t . \nonumber
\end{align}

Although this appears quite loose if $c_t$ is larger than 1, in \Cref{sec:bound_depending_kt} we show that a fine-grained bound in terms of $\min_t c_t$ only results in limited improvements in the analysis. \Cref{eqn:bound_conditional_At} implies the following corollary of our correlation bound. 

\begin{cor}\label{lemma:expected_At_with_condition}
    Let $\Deltac :=  \left( 1 + \frac{(0.5+\kappa)^2}{0.5-\kappa} \right) \cdot \left( \frac{0.5 + \kappa}{0.5 - \kappa} \right)^2$. For any late $(i,t)$ we have $$ \E[A_t \mid t \textup{ arrived}, F_{i,t}] \le \left( 1 - \tau  + \Deltac \cdot (0.5 + \kappa) \right) \cdot c_t . $$
\end{cor}

We are now able to conclude the proof of Lemma~\ref{lemma:approx_late_edges} (ii), as follows. 

\begin{proof}[Proof of \Cref{lemma:approx_late_edges} (ii).]
By \Cref{obs:doneifbetaitatmost1}, it suffices to show that $\rho_{i,t} \ge (0.5 + \kappa)y_{i,t} - (0.5 - \kappa).$ Combining the bound of \Cref{lem:rhoitlowerbound} with \Cref{lemma:expected_At_with_condition} implies 
\begin{align}\label{eqn:rho_lower_bound}
\rho_{i,t} \ge (1 - (0.5 + \kappa) \cdot y_{i,t}) \cdot \left( \tau - \left( 1 - \tau + \Deltac \cdot(0.5 + \kappa) \right) \right) .
\end{align}

For convenience let $g(\kappa) :=  2\tau - 1 - \Deltac \cdot (0.5 + \kappa) $, recalling that $\tau$ is a function of $\kappa$. Then, it suffices to show $(1 - (0.5 + \kappa) \cdot y_{i,t}) \cdot g( \kappa ) \ge (0.5 + \kappa) y_{i,t} - (0.5 - \kappa)$, or equivalently $$g(\kappa) + 0.5 - \kappa \ge (0.5 + \kappa + (0.5 + \kappa) g(\kappa)) \cdot y_{i,t}.$$ For $\kappa = \approxconstant$, we can confirm that the coefficient of $y_{i,t}$ on the right-hand side is positive, and hence it suffices to show this inequality when $y_{i,t} = 1$. This reduces to $$g ( \kappa) \ge \frac{2 \kappa }{0.5 - \kappa}$$ which is readily confirmed by direct computation at $\kappa = \approxconstant$. 
\end{proof}

As a side remark, using \Cref{eqn:rho_lower_bound}, we can observe that for our choice of $\kappa = \approxconstant$, the expectation $\rho_{i,t}$ is bounded away from zero by a constant. In particular, for $\kappa = 0.0115$, we have that $\rho_{i,t} \geq 0.02389$. This can be used to estimate $\rho_{i,t}$ via sampling with small multiplicative error, as we formalize in \Cref{app:sample_based_algo}.

In order to finalize our proof of \Cref{lemma:approx_late_edges} (ii), it only remains to prove our bound on the correlation introduced between offline users, which we do in the following section.

\subsection{Bounding the Correlation --- Proof of \Cref{corollary:bound_correlation_unrestricted}} \label{subsection:correlation_bound}

What remains to conclude the proof of our main \Cref{theorem:main_theorem} is to control the correlation of two users $i$ and $j$ to be free simultaneously, i.e., the bound from \Cref{corollary:bound_correlation_unrestricted}. To this end, we first state and prove \Cref{lem:corrbound} which uses the assumption that $y_{i,t-1}$ and $y_{j,t-1}$ are at most $\tau$.

\begin{lemma}\label{lem:corrbound}
	Define $ \gammac := 1 + \frac{(0.5+\kappa)^2}{0.5-\kappa}$. For any distinct users $i$ and $j$, and any time $t$ such that $y_{i,t-1}, y_{j,t-1} \le \tau$, we have $$\Pr[F_{i,t} \wedge F_{j,t}] \le \gammac \cdot \Pr[F_{i,t}] \cdot \Pr[F_{j,t}].$$
\end{lemma}

To prove this lemma we consider the function $$f(z) := 1 + z \cdot \left( \frac{(0.5+\kappa)^2}{1 - z \cdot (0.5 + \kappa)} \right),$$ which depends on our choice of $\kappa$. Note that $\gammac = f(1)$. For this function, we can prove the following claim.

\begin{restatable}{claim}{claimforcorrbound}\label{claim:claim_for_corrbound}
	For any distinct users $i$ and $j$, and any time $t$ such that $y_{i,t-1}, y_{j,t-1} \le \tau$, we have $$\Pr[F_{i,t} \wedge F_{j,t}] \le f(y_{i,t}) \cdot \Pr[F_{i,t}] \cdot \Pr[F_{j,t}],$$ where $f(z) := 1 + z \cdot \left( \frac{(0.5+\kappa)^2}{1 - z \cdot (0.5 + \kappa)} \right)$. 
\end{restatable}

In order to prove Lemma~\ref{lem:corrbound} from Claim~\ref{claim:claim_for_corrbound}, it suffices to note that $f$ is a monotone increasing function in $[0,1]$, and hence, $f(z) \le f(1) = \gammac$ for all $z \in [0,1]$. 

\begin{proof}[Proof of Claim~\ref{claim:claim_for_corrbound}.]
	We give a proof by induction. As $f(0) = 1$ and all users are available initially, the base case is clear. Assuming the claim is true for fixed $t$, we will prove it for $t+1$ with the assumption $y_{i,t}, y_{j,t} \le \tau$. 

    \paragraph{Proof outline for the inductive step.}
    Our proof proceeds with the following steps:
    
    \begin{enumerate}[label=\text{(S\arabic*)}]
        \item We find an upper bound for the probability that both $i$ and $j$ are not assigned to $t$ via a first proposal conditioned on being free. \label{induction:step_2}
        \item We compute $\Pr[F_{i,t+1}]/\Pr[F_{i,t}] $, in order to apply the inductive hypothesis. \label{induction:step_3}
        \item We apply the induction hypothesis, and use Step~\ref{induction:step_3} to write our bound in terms of $\Pr[F_{i,t+1}] $ and $\Pr[F_{j,t+1}] $. \label{induction:step_4}
        \item We argue that we can upper bound the coefficient in front of $\Pr[F_{i,t+1}] \cdot \Pr[F_{j,t+1}] $ with $f(y_{i,t+1})$. \label{induction:step_5}
    \end{enumerate}

    \paragraph{Step \ref{induction:step_2}: Bounding the probability of not assigning both users via a first proposal.}
    
    As $y_{i,t}, y_{j,t} \le \tau$, they can only be matched as first proposals; hence the probability both $i$ and $j$ are free at time $t+1$ is
    \begin{align}\label{equation:induction_step}
		\Pr[F_{i,t+1} \wedge F_{j,t+1}] & = \Pr[F_{i,t} \wedge F_{j,t}] \cdot  \underbrace{\Pr[(i,t) \notin \mathcal{A}_1 \wedge (j,t) \notin \mathcal{A}_1 \mid F_{i,t} \wedge F_{j,t}]}_{(\star)}  .
    \end{align}
    The first term on the right-hand side of \Cref{equation:induction_step} will later be bounded via the induction hypothesis. The second term $(\star) := \Pr[(i,t) \notin \mathcal{A}_1 \wedge (j,t)  \notin \mathcal{A}_1 \mid F_{i,t} \wedge F_{j,t}]$ can be equivalently written as 
	\begin{align}\label{eqn:induction_prob_i_and_j_alloc}
        (\star) =
		1 - \Pr[(i,t) \in \mathcal{A}_1 \mid F_{i,t} \wedge F_{j,t}] - \Pr[(j,t) &\in \mathcal{A}_1 \mid F_{i,t} \wedge F_{j,t}]  \\
  &+ \Pr[(i,t) \in \mathcal{A}_1 \wedge (j,t) \in \mathcal{A}_1 \mid F_{i,t} \wedge F_{j,t}]. \nonumber
	\end{align}
    Now, observe that $\Pr[(i,t) \in \mathcal{A}_1 \mid F_{i,t} \wedge F_{j,t}] = p_t \cdot \Pr[i \in \textsf{FP}_t] \cdot \alpha_{i,t} = x_{i,t} \cdot \alpha_{i,t}$. The analogous equality holds for $j$. Hence, it remains to get a suitable bound on the joint probability that both users $i$ and $j$ are assigned via a first proposal given they were both free. To this end, we make use of the negative cylinder dependence in pivotal sampling, observing
    \begin{align*}
        \Pr[(i,t) \in \mathcal{A}_1 &\wedge (j,t) \in \mathcal{A}_1 \mid F_{i,t} \wedge F_{j,t}] \\
        & = p_t \cdot \Pr[i \in \textsf{FP}_t \wedge j \in \textsf{FP}_t] \cdot \alpha_{i,t} \cdot \alpha_{j,t} \\ 
        & \le  p_t \cdot \Pr[i \in \textsf{FP}_t] \cdot \Pr[j \in \textsf{FP}_t] \cdot \alpha_{i,t} \cdot \alpha_{j,t} && \text{(Pivotal Sampling Property \ref{pivot_sampling_negative_dependence})}\\ 
        & = p_t \cdot \frac{x_{i,t} \cdot x_{j,t}}{p_t^2} \cdot \alpha_{i,t} \cdot \alpha_{j,t} \\
        & = \frac{x_{i,t} \cdot x_{j,t}}{p_t} \cdot \alpha_{i,t} \cdot \alpha_{j,t}.
    \end{align*}	
    
    Combining all of the above, we can bound the conditional probability that neither $i$ nor $j$ is allocated to $t$ via a first proposal. In other words, the left-hand side of \Cref{eqn:induction_prob_i_and_j_alloc} is at most 
	\begin{align}\label{eqn:induction_step_i_bound}
		\Pr[(i,t) \notin \mathcal{A}_1 \wedge (j,t) \notin \mathcal{A}_1 \mid F_{i,t} \wedge F_{j,t}] & \le 1 - x_{i,t} \alpha_{i,t} - x_{j,t} \alpha_{j,t} + \frac{1}{p_t} \cdot x_{i,t} \alpha_{i,t} x_{j,t}\alpha_{j,t} \\
		&= (1-x_{i,t} \alpha_{i,t})(1-x_{j,t}\alpha_{j,t}) + \left(\frac{1}{p_t} - 1 \right) x_{i,t} \alpha_{i,t} x_{j,t} \alpha_{j,t} . \nonumber
	\end{align}

    \paragraph{Step \ref{induction:step_3}: Comparing $\Pr[F_{i,t+1}]$ to $\Pr[F_{i,t}] $.}
    To prepare for our use of the inductive hypothesis, we compute $\Pr[F_{i,t+1}]/\Pr[F_{i,t}]$ via a straightforward calculation: \begin{align}\label{equation:free_at_tplusone_via_free_at_t}
		\Pr[F_{i,t+1}] &= 1 - (0.5+\kappa) \cdot y_{i,t+1} 
		 = \Pr[F_{i,t}] \cdot \frac{1 - (0.5 + \kappa) \cdot y_{i,t+1}}{1 - (0.5+\kappa) \cdot y_{i,t}}
		= \Pr[F_{i,t}] \cdot \left( 1 - x_{i,t} \cdot \alpha_{i,t} \right).
	\end{align}
	In the final line we used that $(i,t)$ is early. For $j$, we analogously have 
    $$\Pr[F_{j,t+1}] = \Pr[F_{j,t}] \cdot \left( 1 - x_{j,t} \cdot \alpha_{j,t} \right).$$ 
 
    \paragraph{Step \ref{induction:step_4}: Applying the induction hypothesis.}
    Applying the induction hypothesis to \Cref{equation:induction_step}, plugging in Inequality~\eqref{eqn:induction_step_i_bound} and using \Cref{equation:free_at_tplusone_via_free_at_t}, we can bound 
	\begin{align}
		&\Pr[F_{i,t+1} \wedge F_{j,t+1}] \nonumber \\ & = \Pr[F_{i,t} \wedge F_{j,t}] \cdot  \Pr[(i,t) \notin \mathcal{A}_1 \wedge (j,t) \notin \mathcal{A}_1 \mid F_{i,t} \wedge F_{j,t}] \nonumber \\
		& \le \Pr[F_{i,t} \wedge F_{j,t}] \cdot \Big( (1-x_{i,t} \alpha_{i,t})(1-x_{j,t}\alpha_{j,t}) + \left(\frac{1}{p_t} - 1 \right) x_{i,t} \alpha_{i,t} x_{j,t} \alpha_{j,t} \Big) \nonumber \\
		&\le f(y_{i,t}) \cdot \Pr[F_{i,t}] \cdot \Pr[F_{j,t}] \cdot \Big( (1-x_{i,t} \alpha_{i,t})(1-x_{j,t}\alpha_{j,t}) + \left(\frac{1}{p_t} - 1 \right) x_{i,t} \alpha_{i,t} x_{j,t} \alpha_{j,t} \Big) \nonumber \\
		& = f(y_{i,t}) \cdot \Pr[F_{i,t+1}] \cdot \Pr[F_{j,t+1}]  + f(y_{i,t}) \left(\frac{1}{p_t} - 1 \right) \cdot \Pr[F_{i,t}] \cdot \Pr[F_{j,t}] \cdot x_{i,t} \alpha_{i,t} x_{j,t} \alpha_{j,t} . \label{eqn:finallineofstepS3}
	\end{align}
    Here, the first inequality uses Inequality~\eqref{eqn:induction_step_i_bound} from Step~\ref{induction:step_2}, i.e., the upper bound on the probability of both users not being allocated via a first proposal. The second inequality applies the induction hypothesis for $\Pr[F_{i,t} \wedge F_{j,t}] $, and the last equality uses \Cref{equation:free_at_tplusone_via_free_at_t} from Step~\ref{induction:step_3} for both users $i$ and $j$ and rearranges terms. 

    \medskip
    
    We now bound the second summand of \eqref{eqn:finallineofstepS3}, via the following inequality.
	
	\begin{fact}\label{fact:bound_alpha_via_online} For any $(i,t)$ we have \begin{align}\label{inequality:bound_xalpha}
	    x_{i,t} \cdot \alpha_{i,t} \cdot \left( \frac{1}{p_t} - 1 \right) \le (0.5 + \kappa) \cdot \left( 1 - x_{i,t} \alpha_{i,t} \right).
	\end{align}
	\end{fact}
	\begin{proof}
		By Constraint~\eqref{eqn:PPSWConstraint} of the LP, we have that $\frac{1}{p_t} \le \frac{1-y_{i,t}}{x_{i,t}}.$ Thus it suffices to show that $$\alpha_{i,t} (1 - y_{i,t}) - x_{i,t} \alpha_{i,t} \le (0.5 + \kappa) (1 - x_{i,t} \alpha_{i,t})$$ which is equivalent to $$\alpha_{i,t} ( 1 - y_{i,t} - (0.5 -\kappa ) x_{i,t} ) \le 0.5 + \kappa.$$ As $\alpha_{i,t} \le \frac{0.5+\kappa}{1 - (0.5 + \kappa)y_{i,t}}$, the claim follows. 
	\end{proof}
	
	We can apply \Cref{fact:bound_alpha_via_online} to user $j$ and combine it with \Cref{equation:free_at_tplusone_via_free_at_t} in order to bound the second summand via
	\begin{align*}
		& \hspace{0.55cm} f(y_{i,t}) \left(\frac{1}{p_t} - 1 \right) \cdot \Pr[F_{i,t}] \cdot \Pr[F_{j,t}] \cdot x_{i,t} \alpha_{i,t} x_{j,t} \alpha_{j,t} \\ 
		& \le f(y_{i,t}) \cdot \Pr[F_{i,t}] \cdot \Pr[F_{j,t}] \cdot (0.5 + \kappa) \cdot (1 - x_{j,t} \alpha_{j,t}) \cdot x_{i,t} \alpha_{i,t} \\
		& = (0.5 + \kappa) \cdot f(y_{i,t}) \cdot \Pr[F_{j,t+1}] \cdot x_{i,t} \alpha_{i,t} \cdot \Pr[F_{i,t+1}] \cdot (1 - x_{i,t} \cdot \alpha_{i,t} )^{-1} .
	\end{align*}
	Overall, we thus have
	\begin{align*}
		\Pr[F_{i,t+1} \wedge F_{j,t+1}]  
		&\le f(y_{i,t}) \cdot \Pr[F_{i,t+1}] \cdot \Pr[F_{j,t+1}] \cdot \left(1 + (0.5 + \kappa) \cdot \frac{x_{i,t} \alpha_{i,t}}{1 - x_{i,t} \alpha_{i,t}} \right) .
	\end{align*}

    \paragraph{Step \ref{induction:step_5}: Upper bounding the coefficient by $f(y_{i,t+1})$.}
	In order to complete the inductive step, we would like to show that $$f(y_{i,t}) \cdot \left(1 + (0.5 + \kappa) \cdot \frac{x_{i,t} \alpha_{i,t}}{1 - x_{i,t} \alpha_{i,t}} \right) \le f(y_{i,t} + x_{i,t}).$$  
	First, note that as we only consider early pairs, $\alpha_{i,t}$ is always equal to $\frac{0.5 + \kappa }{1 - (0.5 + \kappa) \cdot y_{i,t}}$, so we know $\frac{x_{i,t} \alpha_{i,t}}{1 - x_{i,t} \alpha_{i,t}} = \frac{(0.5 + \kappa) \cdot x_{i,t} }{1 - (0.5 + \kappa) \cdot (x_{i,t} + y_{i,t} )} . $
	Thus to conclude the proof, it suffices to show that $$f(y_{i,t}) \cdot \left(1 + (0.5 + \kappa) \cdot \frac{(0.5 + \kappa) \cdot x_{i,t} }{1 - (0.5 + \kappa) \cdot (x_{i,t} + y_{i,t} )} \right) \le f(y_{i,t} + x_{i,t}).$$
	This is a consequence of our definition $$f(z) := 1 + z \cdot \left( \frac{(0.5+\kappa)^2}{1 - z \cdot (0.5 + \kappa)} \right) .$$ In particular the following claim, whose proof can be found in \Cref{app:proof_inductive_step}, completes the inductive step.
	
	\begin{restatable}{claim}{claiminductivestep} \label{claim:inductive_step} For any $x,y \in [0,1]$ with $x + y \le 1$ and $f(\cdot)$ as stated above, we have $$f(y) \cdot \left( 1 + (0.5 + \kappa) \cdot \frac{(0.5 + \kappa) \cdot x }{1 - (0.5 + \kappa) \cdot (x + y)} \right) \le f(y + x).$$
	\end{restatable}
 
	This concludes the proof of Claim~\ref{claim:claim_for_corrbound}. 
\end{proof}

Now, we can finally prove \Cref{corollary:bound_correlation_unrestricted} which concludes the proof of our main \Cref{theorem:main_theorem}. Let us restate \Cref{corollary:bound_correlation_unrestricted} and prove it afterwards.

\corboundcorrelationunrestricted*

\begin{proof}[Proof of \Cref{corollary:bound_correlation_unrestricted}.]
	We assume that both $y_{i,t} > \tau$ and $y_{j,t} > \tau$; if neither inequality holds the result is clear and follows directly from \Cref{lem:corrbound} while if just one holds the proof proceeds nearly identically with a slightly better guarantee. 
	
	Let $t^i$ denote the latest resource in $[T]$ such that $y_{i, t^i-1} \le \tau$ and $y_{i, t^i} > \tau$ and similarly let $t^j$ denote the latest resource in $[T]$ such that $y_{j,t^j-1} \le \tau$ and $y_{i,t^j} > \tau$. 
	
	Let $A_i$ denote the event that $i$ is allocated to some arrival in $[t^i, t-1]$ and let $A_j$ denote the event that $j$ is allocated to some arrival in $[t^j, t-1]$. By the hypothesis that $\Pr[(i,t') \in \mathcal{A}_1] + \Pr[(i,t') \in \mathcal{A}_2] = (0.5 + \kappa) \cdot x_{i,t'}$ for all $t' < t$, we have $$\Pr[A_i] = \sum_{t' \in [t^i, t-1]} (0.5 + \kappa) \cdot x_{i, t'} = (0.5 + \kappa) \cdot (y_{i,t} - y_{i,t^i}) \le (0.5 + \kappa) \cdot (1 - \tau) = 2\kappa.$$ An analogous upper bound holds for $\Pr[A_j]$.

    To simplify notation, let us assume for a moment that $i^j \leq i^{j'}$ (if $i^j > i^{j'}$, simply swap the roles of $j$ and $j'$ in the following line). We apply \Cref{lem:corrbound} to get
	\begin{align*}
		\Pr[F_{i,j} \wedge F_{i,j'}] &\le \Pr[F_{i^j,j} \wedge F_{i^{j'},j'}] \nonumber \\
		& = \Pr[F_{i^j,j} \wedge F_{i^{j},j'}] \cdot \Pr[F_{i^j,j} \wedge F_{i^{j'},j'} \mid F_{i^j,j} \wedge F_{i^{j},j'} ] \nonumber \\
		&\le  \gammac \cdot \Pr[F_{i^j,j}] \cdot \Pr[F_{i^{j},j'}] \cdot \Pr[F_{i^j,j} \wedge F_{i^{j'},j'} \mid F_{i^j,j} \wedge F_{i^{j},j'} ]  && \text{(via \Cref{lem:corrbound})} \nonumber \\
		& = \gammac \cdot \Pr[F_{i^j,j}] \cdot \Pr[F_{i^{j},j'}] \cdot \Pr[F_{i^{j'},j'} \mid F_{i^j,j} \wedge F_{i^{j},j'} ]\enspace.  \nonumber \\
	\end{align*}
	In this expression, we aim to combine the last two factors concerning the events if item $j'$ is free at some point in time. To this end, observe that
	\begin{align*}
		 \Pr[F_{i^{j},j'}] \cdot \Pr[F_{i^{j'},j'} \mid F_{i^j,j} \wedge F_{i^{j},j'} ] & =  \Pr[F_{i^{j},j'}] \cdot \prod_{i' = i^j}^{i^{j'} -1 } \left( 1 - q_{i'} \cdot \Pr[j' \in \mathsf{FP}_{i'}  ] \cdot \alpha_{i',j'}  \right) \\
		 & =  \Pr[F_{i^{j},j'}] \cdot \prod_{i' = i^j}^{i^{j'} -1 } \left( 1 - x_{i',j'} \cdot \alpha_{i',j'}  \right) \\
		 & = \Pr[F_{i^{j'},j'}] \enspace,
	\end{align*}
	where the last equality uses the same ideas as Step~\ref{induction:step_3} in the proof of \Cref{claim:claim_for_corrbound}. So, overall, we have 
	\begin{align}\label{inequality:correlationboundlate}
		\Pr[F_{i,j} \wedge F_{i,j'}] \leq  \Pr[F_{i^j,j} \wedge F_{i^{j'},j'}] \leq \gammac \cdot \Pr[F_{i^j,j}] \cdot \Pr[F_{i^{j'},j'}] \enspace. 
	\end{align}
	With this in mind, we are ready to prove the final statement as 
	\begin{align*}
		\Pr[F_{i,j} \wedge F_{i,j'}] &\le \Pr[F_{i^j,j} \wedge F_{i^{j'},j'}] \\
		&\le  \gammac \cdot \Pr[F_{i^j,j}] \cdot \Pr[F_{i^{j'},j'}]  && \text{(via \Cref{inequality:correlationboundlate})}\\
		&= \gammac \cdot \left( \Pr[F_{i,j}] + \Pr[A_j] \right) \cdot \left( \Pr[F_{i,j'}] + \Pr[A_{j'}] \right) \\
		&\le \gammac \cdot \left( \Pr[F_{i,j}] + 2\kappa \right) \cdot \left( \Pr[F_{i,j'}] + 2\kappa \right) \\
		&\le \gammac \cdot \left( 1 + \frac{4\kappa}{0.5 - \kappa} + \frac{4\kappa^2}{(0.5-\kappa)^2} \right) \cdot \Pr[F_{i,j}] \cdot \Pr[F_{i,j'}] \\
		& = \gammac \cdot \left( \frac{0.5+\kappa}{0.5-\kappa} \right)^2 \cdot \Pr[F_{i,j}] \cdot \Pr[F_{i,j'}] \\
		& = \Deltac \cdot \Pr[F_{i,j}] \cdot \Pr[F_{i,j'}]  ,
	\end{align*}
	where in the last inequality we used $\Pr[F_{i,j}], \Pr[F_{i,j'}] \ge 0.5 - \kappa$ and the last equality applies $\gammac := 1 + \nicefrac{(0.5+\kappa)^2}{0.5-\kappa}$. 
\end{proof}

\section{Analyzing the Sample-based Algorithm}
\label{app:sample_based_algo}

To update \Cref{allocationalgexact} to run in polynomial time, instead of computing the exact value of $\rho_{i,t}$ we estimate it with polynomially many samples. For simplicity, we present the algorithm and its analysis for Bernoulli arrivals when every success probability $q_{i,t}$ equals 1 (the relevant changes needed for the generalizations are described in \Cref{app:beyond_bernoulli} and \Cref{app:stochasticrewards}, respectively). The pseudocode is presented below; observe that we reduce the constant $\kappa$ by an arbitrarily small $\epsilon > 0$ in \Cref{line:sample_based:sample_first}.

\begin{algorithm}[H]
	\caption{(parametrized by $\epsilon > 0$)}
	\label{allocationalgwithsampling}
	\begin{algorithmic}[1]

		\State $\kappa \gets \approxconstant - \epsilon$ \label{line:sample_based:sample_first}
		\State Solve \eqref{LP} for $\{x_{i,t} \}$  \label{line:sample_based:sample_solveLP}
		\For{each time $t$, if $t$ arrives} \Comment{w.p. $p_t$}
		
		\State Define users $\FPt := \textsf{PS}( ( x_{i,t}/p_t )_{i \in I})$  \Comment{at most $c_t$ users get first proposals}  \label{line:sample_based:sample_firstproposal}
		\For{each user $i \in \FPt$}
		\If{$i$ is available}
		\State Allocate $i$ to $t$ with probability $\alpha_{i,t} := \min \left(1, \frac{0.5 + \kappa }{1 - (0.5 + \kappa) \cdot \sum_{t' < t} x_{i,t'}} \right)$  \label{line:sample_based:sample_alphait} \label{line:sample_based:sample_firstmatch}
		\EndIf 
		
		\EndFor
		
		\State Let $A_t \gets \text{number of users allocated to } t \text{ thus far}$ \label{line:sample_based:sample_defAt}
  
		\State Define users $\SPt:= \textsf{PS}( ( (1 - \frac{A_t}{c_t}) \cdot x_{i,t}/p_t )_{i \in I})$ \label{line:sample_based:sample_secondproposal} \Comment{$\le c_t - A_t$ users get second proposal} 
		\For{each user $i \in \SPt$ with $\alpha_{i,t} = 1$}
		\If{$i$ is available} \label{line:sample_based:sample_ifiavail}
        \State Do not compute $\sigma_{i,t} := \E[ \mathbbm{1}[i \text{ available after \Cref{line:sample_based:sample_defAt}}] \cdot (1 - \frac{A_t}{c_t}) \mid t \text{ arrived}, (\hat{\sigma}_{i,t'})_{t' < t}] $; 
        \State Instead compute $\hat{\sigma}_{i,t}  \gets $ Empirical average of $\sigma_{i,t}$ over $N := 50 n T \cdot ( \epsilon / 400T)^{-2} \cdot \kappa^{-2} $ independent simulations, using previously computed values $(\hat{\sigma}_{i,t'})_{t' < t}$
        \State $\hat{\beta}_{i,t} \gets \min \Big(1, \left( (0.5 + \kappa) \cdot \sum_{t' < t} x_{i,t'} - (0.5 - \kappa) \right) \cdot \frac{1}{\hat{\sigma}_{i,t}} \Big).$
		\State Allocate $i$ to $t$ with prob. $\hat{\beta}_{i,t}$ \label{line:sample_based:sample_secondmatch}
		\EndIf
		\EndFor
		
		\EndFor
	\end{algorithmic}
\end{algorithm}	

As before, the definition of $\rho_{i,t}$ is over the randomness in the arrivals and algorithm up to when it reaches \Cref{line:sample_based:sample_defAt} for arrival $t$ in \Cref{allocationalgexact}, with the previously computed values of $(\hat{\sigma}_{i,t'})$ for $t' < t$. In particular, we do not recalculate these, but rather inductively use them as defined previously. This is why we use the shorthand of ``conditioning on $(\hat{\sigma}_{i, t'})_{t' < t}$'' when defining $\rho_{i,t}$.

We start with the observation that our algorithm is unchanged for early pairs $(i,t)$. In particular, the following lemmas still hold for \Cref{allocationalgwithsampling}. 

\lemmaapproxearlyedges*

\claimforcorrbound*

In the remainder of the analysis, we will need to track the errors incurred by sampling. 
Note that by the Chernoff-Hoeffding bound, if $\sigma_{i,t}$ is bounded away from 0 then the empirical average $\hat{\sigma}_{i,t}$ will be within a close multiplicative factor. 

\begin{obs}\label{obs:chernoff}
If $\sigma_{i,t} \ge \kappa$ then we have that $ \sigma_{i,t} / \hat{\sigma}_{i,t} \in \left[ 1 - \frac{ \epsilon }{ 200 T}, 1 + \frac{\epsilon}{200 T} \right]$ with probability at least $1 - 2 \cdot \exp(-100 n T).$
\end{obs}

\begin{proof}
    We straightforwardly bound \begin{align*}
        \Pr \left[ |\hat{\sigma}_{i,t} - \sigma_{i,t}| \ge  \frac{\epsilon}{400 T} \cdot \sigma_{i,t} \right] &\le  2 \cdot \exp \left( -2 \cdot N \cdot ((\epsilon / 400T) \cdot \sigma_{i,t})^2 \right)  \\
        &\le 2 \cdot \exp \left( -2 \cdot N \cdot (( \epsilon / 400T) \cdot \kappa )^2 \right) \\
        &\le 2 \cdot \exp \left( -100 n T \right). 
        \end{align*}
    Thus with probability at least $1 - 2 \cdot \exp(-100nT)$, we have $$\sigma_{i,t}/\hat{\sigma}_{i,t} \in [(1 + \epsilon/400T)^{-1}, (1 - \epsilon / 400 T)^{-1}].$$ The observation follows directly. 
\end{proof}

We now show inductively that our algorithm allocates each $(i,t)$ with probability close to the idealized value of $(0.5 + \kappa ) \cdot x_{i,t}$ from the exact (exponential-time) calculations. In particular, we show that that we achieve a value of $(0.5 + \kappa \pm \epsilon_t) \cdot x_{i,t}$ where the error $\epsilon_t$ accumulates only linearly in $t$. 

\begin{lemma}\label{lemma:samplingapprox}
    For any online arrival $t$, with probability at least $1 - 2nt \cdot \exp(-100nT) $, we have for every $t' \le t$ that 
    \begin{align}
    \Pr[(i,t') \in \mathcal{A}_1] + \Pr[(i,t') \in \mathcal{A}_2] \in [(0.5 + \kappa - \epsilon \cdot t'/T) \cdot x_{i,t'}, (0.5 + \kappa + \epsilon \cdot t'/T) \cdot x_{i,t'}].
    \end{align}
\end{lemma}

Note that once we have \Cref{lemma:samplingapprox}, it is immediate to bound the gain of \Cref{allocationalgwithsampling}. In particular, the social welfare achieved by \Cref{allocationalgwithsampling} is with probability at least $1 - 2nT \cdot \exp(-100nT)$ lower-bounded by 
\begin{align*}
\sum_t \sum_i (0.5 + \kappa - \epsilon \cdot t/T) \cdot x_{i,t} \cdot v_{i,t} &\ge \sum_t \sum_i (0.5 + \kappa - \epsilon) \cdot x_{i,t} \cdot v_{i,t} \\
&= (0.5115 - 2 \epsilon) \cdot \text{OPT}\eqref{LP} \\
&\ge (0.5115 - 2 \epsilon) \cdot \opton .
\end{align*}
Note that for a realization of \Cref{allocationalgwithsampling}, we can estimate its gain within a small multiplicative error factor by simulating it over polynomially-many independently sampled arrival sequences. Thus, this guarantee can be obtained with high probability, and it only remains to prove \Cref{lemma:samplingapprox}.

\begin{proof}[Proof of \Cref{lemma:samplingapprox}.]
    By induction on $t$. We consider only the case where the lemma's statement holds for all $\{1, 2, \ldots, t-1\}$, and note this is with probability at least $1 - 2n (t-1) \cdot \exp(-100nT)$ by the inductive hypothesis. Note that for any $i$ such that $(i,t)$ is early, we are done by \Cref{lemma:approx_early_edges}.
    
    For convenience of notation, let $\epsilon_t := \epsilon \cdot t/T$ denote the error accumulated up to time $t$. Using this notation, we can apply the inductive hypothesis to bound 
    \begin{align}
    \Pr[F_{i,t}] \in [ 1 - (0.5 + \kappa + \epsilon_t) \cdot y_{i,t}, 1 -  (0.5 + \kappa - \epsilon_t) \cdot y_{i,t}] .
    \end{align}
    
     Hence the probability late $(i,t)$ is allocated as a first pick satisfies $$\Pr[(i,t) \in \mathcal{A}_1]  = x_{i,t} \cdot \Pr[F_{i,t}] \in [ x_{i,t} - x_{i,t} \cdot (0.5 + \kappa + \epsilon_t ) \cdot y_{i,t}, x_{i,t} -  x_{i,t} \cdot (0.5 + \kappa - \epsilon_t ) \cdot y_{i,t}]$$ where we used the induction hypothesis for bounding $\Pr[F_{i,t}]$. 
     
    By Equation~\eqref{eqn:rhotimesbeta} the probability $(i,t)$ is allocated as a second pick is given by 
    \begin{align}
    \Pr[(i,t) \in \mathcal{A}_2] = x_{i,t} \cdot \sigma_{i,t} \cdot \hat{\beta}_{i,t}. \label{eqn:samplingsecondpick}
    \end{align}
    
    As before, we aim to show that $\sigma_{i,t}$ is bounded away from $0$. 
    Note that analogously to \Cref{eqn:rhoitlowerbound}, we have 
    \begin{align}
    \sigma_{i,t} &\ge \Pr[F_{i,t}] \cdot \left( \tau - \frac{\E[A_t \mid t \text{ arrived}, F_{i,t}]}{c_t}    \right) \nonumber \\
    &\ge \left( 1 - ((0.5 + \kappa + \epsilon_t) \cdot y_{i,t}) \right)  \cdot \left( \tau - \frac{\E[A_t \mid t \text{ arrived}, F_{i,t}]}{c_t}    \right). \label{sigmaitlowerbound} 
    \end{align}
    To bound the conditional expectation $\E[A_t \mid t \text{ arrived}, F_{i,t}]$, we will (as before) upper bound the joint probability $\Pr[F_{i,t} \wedge F_{j,t}]$, analogously to \Cref{corollary:bound_correlation_unrestricted}. Here, the main contribution is from \Cref{claim:claim_for_corrbound}; the probability mass from late edges does not greatly affect it for small $\kappa$, even when taking into account the possible error introduced by sampling. As our algorithm is unchanged along early edges, the proof from the body of the paper goes through in a very similar fashion, which we formalize below. 

    We assume that both $y_{i,t} > \tau$ and $y_{j,t} > \tau$ (if neither, or just one of these inequalities holds, the proof proceeds nearly identically with better bounds).  Let $t^i$ denote the latest resource in $[T]$ such that $y_{i, t^i-1} \le \tau$ and $y_{i, t^i} > \tau$ and similarly let $t^j$ denote the latest resource in $[T]$ such that $y_{j,t^j-1} \le \tau$ and $y_{i,t^j} > \tau$. Let $A_i$ denote the event that $i$ is allocated to some arrival in $[t^i, t-1]$ and let $A_j$ denote the event that $j$ is allocated to some arrival in $[t^j, t-1]$. Using the hypothesis that $\Pr[(i,t') \in \mathcal{A}_1] + \Pr[(i,t') \in \mathcal{A}_2] \le (0.5 + \kappa + \epsilon_{t'}) \cdot x_{i,t'}$ for all $t' < t$, we have 
    \begin{align*}
    \Pr[A_i] & \leq \sum_{t' \in [t^i, t-1]} (0.5 + \kappa + \epsilon_{t'}) \cdot x_{i, t'} \\
    &\le (0.5 + \kappa + \epsilon_t) \cdot (y_{i,t} - y_{i,t^i}) \\
    &\le (0.5 + \kappa + \epsilon_t) \cdot (1 - \tau) \\
    & = (0.5 + \kappa + \epsilon_t ) \cdot \frac{2 \kappa}{0.5+ \kappa} \\
    & = 2 \kappa + \epsilon_t \cdot \frac{2 \kappa}{0.5+\kappa} 
    \end{align*} 
    An analogous upper bound holds for $\Pr[A_j]$. For convenience, let us define $\etac := \nicefrac{2 \kappa}{0.5+\kappa}$. With this, we can bound 
	\begin{align*}
		\Pr[F_{i,t} \wedge F_{j,t}] &\le \Pr[F_{i,t^i} \wedge F_{j,t^j}] \\
		&\le  \gammac \cdot \Pr[F_{i,t^i}] \cdot \Pr[F_{j,t^j}]  \hspace{5.5cm} \text{(\Cref{lem:corrbound})}\\
		&= \gammac \cdot \left( \Pr[F_{i,t}] + \Pr[A_i] \right) \cdot \left( \Pr[F_{j,t}] + \Pr[A_j] \right) \\
        &\le \gammac \cdot \left( \Pr[F_{i,t}] + 2 \kappa + \epsilon_t \etac \right) \cdot \left( \Pr[F_{j,t}] + 2 \kappa + \epsilon_t \etac \right) \\
        & \le \zeta_{\kappa, \epsilon_t} \cdot \Pr[F_{i,t}] \cdot \Pr[F_{j,t}]
    \end{align*}
    where in the last inequality, we first use a lower bound on $\Pr[F_{i,t}]$ and $\Pr[F_{j,t}]$ of $0.5 - \kappa - \epsilon_t$ and defined $$\zeta_{\kappa, \epsilon_t} := \gamma_{\kappa} \cdot \left( 1 + \frac{2(2\kappa + \epsilon_t \etac)}{0.5 - \kappa - \epsilon_t} + \frac{(2 \kappa + \epsilon_t \etac )^2}{(0.5 - \kappa - \epsilon_t)^2} \right) .$$ 

    Now, following the calculation of \Cref{eqn:AtconditionedtarrivedFit}, we have 
    \begin{align*}
    \E[A_t \mid t \text{ arrived}, F_{i,t}]   &= \frac{x_{i,t}}{p_t} + \sum_{j \neq i} \frac{\Pr[F_{i,t} \wedge F_{j,t}]}{\Pr[F_{i,t}]} \cdot \frac{x_{j,t}}{p_t} \cdot \alpha_{j,t} \\
    &\le \frac{x_{i,t}}{p_t} + \sum_{j \neq i} \zeta_{\kappa, \epsilon_t} \cdot \Pr[F_{j,t}] \cdot \frac{x_{j,t}}{p_t} \cdot \alpha_{j,t} \\
    &\le \frac{x_{i,t}}{p_t} + \zeta_{\kappa, \epsilon_t} \cdot (0.5 + \kappa + 2 \epsilon_t) \cdot c_t. 
    \end{align*}
    For the final inequality, we are using $\Pr[F_{j,t}] \le 1 - (0.5 + \kappa - \epsilon_t) \cdot y_{i,t}$ by our hypothesis  and substituting ${\alpha_{j,t} \le \frac{0.5 + \kappa}{1 - (0.5 + \kappa) y_{i,t}}}.$
    Using that $\nicefrac{x_{i,t}}{p_t} \le 1 - \tau$ as $(i,t)$ is late, we have 
    \begin{align}\label{eqn:sample_bound_At_conditioned}
        \E[A_t \mid t \text{ arrived}, F_{i,t}] \le \left( 1 - \tau + \zeta_{\kappa, \epsilon_t} \cdot (0.5 + \kappa + 2 \epsilon_t ) \right) \cdot c_t
    \end{align}
    as in \Cref{lemma:expected_At_with_condition}.

    Now, starting from \Cref{sigmaitlowerbound} and using \Cref{eqn:sample_bound_At_conditioned}, we note 
    \begin{align*}
    \sigma_{i,t} &\ge \left( 1 - ((0.5 + \kappa + \epsilon_t) \cdot y_{i,t}) \right)  \cdot \left( 2 \tau - 1 - \zeta_{\kappa, \epsilon_t} \cdot (0.5 + \kappa + 2 \epsilon_t ) )\right) \\
    & \ge \left( 0.5 - \kappa - \epsilon_t \right)  \cdot \left( 2 \tau - 1 - \zeta_{\kappa, \epsilon_t} \cdot (0.5 + \kappa + 2 \epsilon_t ) )\right) \\
    &\ge (0.5 - \kappa - \epsilon_t ) \cdot \left( \frac{2 \kappa}{0.5 - \kappa} + \epsilon_t \right) && \text{for } \epsilon_t \le 0.0001 \\
    &\ge 2 \kappa  + 0.1\epsilon_t  ,
    \end{align*}
    where the second inequality uses $y_{i,t} \le 1$ and the last inequality is a straightforward calculation for sufficiently small $\epsilon$. The third inequality is calculation-heavy and holds only for small $\epsilon$ and $\kappa \le \approxconstant - \epsilon$, and requires some slightly tedious calculations. For example, we can upper bound $\zeta_{\kappa, \epsilon_t}$ by noting that for $\epsilon_t$ sufficiently small $\frac{2(2\kappa + \epsilon_t \eta_k)}{0.5 - \kappa - \epsilon_t} \le \frac{2(2 \kappa)}{0.5 - \kappa} + 0.5\epsilon_t$ and $\frac{(2 \kappa + \epsilon_t \etac )^2}{(0.5 - \kappa - \epsilon_t)^2} \le \frac{(2 \kappa)^2}{(0.5 - \kappa)^2} + 0.1\epsilon_t$. 

    Thus 
    \begin{align*}
    \zeta_{\kappa, \epsilon_t} \cdot (0.5 + \kappa + 2 \epsilon_t) &\le \gamma_{\kappa} \cdot \left( 1 + \frac{2(2 \kappa)}{0.5 - \kappa} + \frac{(2 \kappa)^2}{(0.5 - \kappa)^2} + 0.6 \epsilon_t \right) \cdot (0.5 + \kappa + 2 \epsilon_t) \\
    &\le \gamma_{\kappa} \cdot \left( 1 + \frac{2(2 \kappa)}{0.5 - \kappa} + \frac{(2 \kappa)^2}{(0.5 - \kappa)^2}   \right) \cdot (0.5 + \kappa ) +  \epsilon_t \cdot (\dagger) + 1.2 \cdot \gamma_{\kappa} \cdot \epsilon_t^2
    \end{align*}
    for $(\dagger) := \gamma_{\kappa} \cdot \left( 0.6(0.5 + \kappa) + 2 \left(1 + \frac{2(2 \kappa)}{0.5 - \kappa} + \frac{(2 \kappa)^2}{(0.5 - \kappa)^2} \right) \right) < 4.$ We loosely bound $4 \epsilon_t + 4 \epsilon_t^2 \le 5 \epsilon_t$ for $\epsilon$ sufficiently small. So we just need to show  $2\tau - 1 -  \gamma_{\kappa} \cdot \left( 1 + \frac{2(2 \kappa)}{0.5 - \kappa} + \frac{(2 \kappa)^2}{(0.5 - \kappa)^2}   \right) \cdot (0.5 + \kappa) - 5\epsilon_t \ge \frac{2 \kappa}{0.5 - \kappa} + \epsilon_t.$ Using that $\epsilon_t \le \epsilon$ it suffices to show $6 \epsilon \le 2\tau - 1 -  \gamma_{\kappa} \cdot \left( 1 + \frac{2(2 \kappa)}{0.5 - \kappa} + \frac{(2 \kappa)^2}{(0.5 - \kappa)^2}   \right) \cdot (0.5 + \kappa ) - \frac{2 \kappa}{0.5 - \kappa}.$ Recalling that $\kappa := \approxconstant -\epsilon$, we note this reduces to a single-variable inequality in only $\epsilon$. This is not easy to show directly, as it crucially is true for the magic constant $\approxconstant$, but can readily be shown by computer verification. Indeed, the RHS and LHS are easily seen to be 100-Lipschitz as functions of $\epsilon \in [0, 0.1]$, say, so we confirm the RHS is at least $10^{-5}$ larger than the LHS on a grid of $10^6$ points on $[0,0.1]$. 
    
    Hence, we get $\sigma_{i,t}$ is bounded away from $0$ and can apply \Cref{obs:chernoff}: for any fixed $i$ such that $(i,t)$ is late, we have with probability at least $1 - 2 \cdot \exp(-100nT)$ that $\sigma_{i,t}/ \hat{\sigma}_{i,t} \in [1 - \epsilon/200T, 1 + \epsilon / 200T].$ Note that  in this case we have
    \begin{align*}
    \frac{1}{\hat{\sigma}_{i,t}} &\le \frac{1 + \epsilon / 200T}{\sigma_{i,t}} \\
    &\le \frac{1 + \epsilon / 200T}{2 \kappa + 0.1\epsilon_t} .
    \end{align*}

    Recall $\hat{\beta}_{i,t} := \min \Big(1, \left( (0.5 + \kappa) \cdot y_{i,t} - (0.5 - \kappa) \right) \cdot \frac{1}{\hat{\sigma}_{i,t}} \Big).$ Note 
    \begin{align*}
    \left( (0.5 + \kappa) \cdot y_{i,t} - (0.5 - \kappa) \right) \cdot \frac{1}{\hat{\sigma}_{i,t}} \le 2 \kappa \cdot  ( 2 \kappa + 0.1 \epsilon_t)^{-1} \cdot (1 + \epsilon / 200T)  \le 1
    \end{align*}
    where the final (loose) inequality follows as $T \ge 1$ and $\kappa \le 0.4$. This implies
    \begin{align*}
    \Pr[(i,t) \in \mathcal{A}_2] &= x_{i,t} \cdot \sigma_{i,t} \cdot \hat{\beta}_{i,t} &&\text{(\Cref{eqn:samplingsecondpick})} \\
    &= x_{i,t} \cdot \sigma_{i,t} \cdot \left( (0.5 + \kappa) \cdot y_{i,t} - (0.5 - \kappa) \right) \cdot \frac{1}{\hat{\sigma}_{i,t}} \\
    &\ge x_{i,t} \cdot (1 - \epsilon / 200T) \cdot \left( (0.5 + \kappa) \cdot y_{i,t} - (0.5 - \kappa) \right)   
    \end{align*}
    and similarly $$\Pr[(i,t) \in \mathcal{A}_2] \le x_{i,t} \cdot (1 + \epsilon / 200T) \cdot \left( (0.5 + \kappa) \cdot y_{i,t} - (0.5 - \kappa) \right)  .$$

Then, we have 
\begin{align*}
\Pr[(i,t) \in \mathcal{A}_1] &+ \Pr[(i,t) \in \mathcal{A}_2] \\
&\ge x_{i,t} \cdot (1 - (0.5 + \kappa + \epsilon_t) \cdot y_{i,t} + (1 - \nicefrac{\epsilon}{200T}) \cdot ((0.5 + \kappa) y_{i,t} - (0.5 - \kappa)) \\
& = x_{i,t} \cdot \left(  1 + y_{i,t} \left( -0.5 - \kappa - \epsilon_t + (1 - \nicefrac{\epsilon}{200T}) \cdot (0.5 + \kappa) \right) - (1 - \nicefrac{\epsilon}{200T}) \cdot (0.5-\kappa) \right) \\
& \geq x_{i,t} \cdot \left(  1 + \left( -0.5 - \kappa - \epsilon_t + (1 - \nicefrac{\epsilon}{200T}) \cdot (0.5 + \kappa) \right) - (1 - \nicefrac{\epsilon}{200T}) \cdot (0.5-\kappa) \right) ,
\end{align*}
where the last inequality uses that the coefficient of $y_{i,t}$ above is $-0.5 - \kappa - \epsilon_t + (1 - \nicefrac{\epsilon}{200T} )(0.5 + \kappa)$, which is non-positive. Hence we can bound 
\begin{align*}
\Pr[(i,t) \in \mathcal{A}_1] + \Pr[(i,t) \in \mathcal{A}_2] 
&\ge x_{i,t} \cdot (1 - (0.5 + \kappa + \epsilon_t)   + (1 - \epsilon / 200T) \cdot 2 \kappa ) \\
&= x_{i,t} \cdot \left( 0.5 + \kappa  - \epsilon_t - \frac{\epsilon}{100T} \cdot \kappa \right). \\
&\ge x_{i,t} \cdot \left( 0.5 + \kappa  -\epsilon_{t+1} \right).
\end{align*} We also have the analogous upper bound 
\begin{align*}
\Pr[(i,t) \in \mathcal{A}_1] + \Pr[(i,t) \in \mathcal{A}_2] \le x_{i,t} \cdot \left( 0.5 + \kappa + \epsilon_{t+1} \right).
\end{align*}
By the union bound, with probability at least $1 - 2n \cdot \exp(-100nT)$ these two bounds hold for all $i$ with $(i,t)$ late. Via the inductive hypothesis, our starting assumption occurred with probability at least $1 - 2n(t-1) \cdot \exp(-100nT)$. Hence, by a final application of the union bound, we have that our desired property for arrivals $\{1, 2, \ldots, t\}$ holds with probability at least $1 - 2nt \cdot \exp(-100nT).$
    
\end{proof}
\section{Conclusion and Future Directions}
\label{sec:conclusion}

We gave the first algorithm achieving an approximation ratio strictly better than $\nicefrac{1}{2}$ for capacitated online resource allocation, when comparing to the (computationally inefficient) optimum online algorithm. Our algorithm crucially limited the (necessary) positive correlation between offline users, and analyzed this via an inductive bound depending on the total LP flow sent to an individual user. This challenge does not arise in competitive analysis, and lends credence to the value of the optimum online as a complementary benchmark to the prophet.

Numerous directions for future research are suggested by our work. Can our guarantee of $0.5 + \kappa$ for $\kappa = \approxconstant$ be improved, perhaps by rounding stronger LPs? Is there a better tradeoff possible between the amount of positive correlation we introduce for early arrivals and the approximation ratio possible on late ones? 

Finally, we believe the techniques developed for handling positive correlation may prove useful for future generalizations. The prophet inequalities literature has studied more general settings than capacitated allocation where the tight $\nicefrac{1}{2}$-guarantee is known \cite{feldman2015combinatorial, dutting2020prophet}, and our work gives some evidence that it is possible to get an improved approximation ratio against the online benchmark for these problems as well. 
\newpage
\bibliographystyle{alpha}
\bibliography{abb,a,ultimate}

\newcommand{\etalchar}[1]{$^{#1}$}
\begin{thebibliography}{ACCB{\etalchar{+}}23}

\bibitem[ACCB{\etalchar{+}}23]{avadhanula2023dynamicocrs}
Vashist Avadhanula, Andrea Celli, Riccardo Colini-Baldeschi, Stefano Leonardi, and Matteo Russo.
\newblock Fully dynamic online selection through online contention resolution schemes.
\newblock In {\em Proceedings of the Thirty-Seventh AAAI Conference on Artificial Intelligence and Thirty-Fifth Conference on Innovative Applications of Artificial Intelligence and Thirteenth Symposium on Educational Advances in Artificial Intelligence}, AAAI'23/IAAI'23/EAAI'23. AAAI Press, 2023.

\bibitem[Ada11]{adamczyk2011improved}
Marek Adamczyk.
\newblock Improved analysis of the greedy algorithm for stochastic matching.
\newblock {\em Information Processing Letters (IPL)}, 111(15):731--737, 2011.

\bibitem[AGKM11]{aggarwal2011online}
Gagan Aggarwal, Gagan Goel, Chinmay Karande, and Aranyak Mehta.
\newblock Online vertex-weighted bipartite matching and single-bid budgeted allocations.
\newblock In {\em Proceedings of the 22nd Annual ACM-SIAM Symposium on Discrete Algorithms (SODA)}, pages 1253--1264, 2011.

\bibitem[AGM15]{adamczyk2015improved}
Marek Adamczyk, Fabrizio Grandoni, and Joydeep Mukherjee.
\newblock Improved approximation algorithms for stochastic matching.
\newblock In Nikhil Bansal and Irene Finocchi, editors, {\em Algorithms - {ESA} 2015 - 23rd Annual European Symposium, Patras, Greece, September 14-16, 2015, Proceedings}, volume 9294 of {\em Lecture Notes in Computer Science}, pages 1--12. Springer, 2015.

\bibitem[AHL12]{alaei2012online}
Saeed Alaei, MohammadTaghi Hajiaghayi, and Vahid Liaghat.
\newblock Online prophet-inequality matching with applications to ad allocation.
\newblock In {\em Proceedings of the 13th ACM Conference on Electronic Commerce (EC)}, pages 18--35, 2012.

\bibitem[AHL13]{alaei2013gap}
Saeed Alaei, MohammadTaghi Hajiaghayi, and Vahid Liaghat.
\newblock The online stochastic generalized assignment problem.
\newblock In Prasad Raghavendra, Sofya Raskhodnikova, Klaus Jansen, and Jos{\'{e}} D.~P. Rolim, editors, {\em Approximation, Randomization, and Combinatorial Optimization. Algorithms and Techniques - 16th International Workshop, {APPROX} 2013, and 17th International Workshop, {RANDOM} 2013, Berkeley, CA, USA, August 21-23, 2013. Proceedings}, volume 8096 of {\em Lecture Notes in Computer Science}, pages 11--25. Springer, 2013.

\bibitem[Ala14]{alaei2014bayesian}
Saeed Alaei.
\newblock Bayesian combinatorial auctions: Expanding single buyer mechanisms to many buyers.
\newblock {\em SIAM Journal on Computing (SICOMP)}, 43(2):930--972, 2014.

\bibitem[AM23]{aouad2023nonparametric}
Ali Aouad and Will Ma.
\newblock A nonparametric framework for online stochastic matching with correlated arrivals.
\newblock In Kevin Leyton{-}Brown, Jason~D. Hartline, and Larry Samuelson, editors, {\em Proceedings of the 24th {ACM} Conference on Economics and Computation, {EC} 2023, London, United Kingdom, July 9-12, 2023}, page 114. {ACM}, 2023.

\bibitem[ANSS19]{anari2019nearly}
Nima Anari, Rad Niazadeh, Amin Saberi, and Ali Shameli.
\newblock Nearly optimal pricing algorithms for production constrained and laminar bayesian selection.
\newblock In {\em Proceedings of the 20th ACM Conference on Economics and Computation (EC)}, pages 91--92, 2019.

\bibitem[BC21]{blanc2021multiway}
Guy Blanc and Moses Charikar.
\newblock Multiway online correlated selection.
\newblock In {\em Proceedings of the 62nd Symposium on Foundations of Computer Science (FOCS)}, pages 1277--1284, 2021.

\bibitem[BDL22]{braverman2022max}
Mark Braverman, Mahsa Derakhshan, and Antonio~Molina Lovett.
\newblock Max-weight online stochastic matching: Improved approximations against the online benchmark.
\newblock In David~M. Pennock, Ilya Segal, and Sven Seuken, editors, {\em {EC} '22: The 23rd {ACM} Conference on Economics and Computation, Boulder, CO, USA, July 11 - 15, 2022}, pages 967--985. {ACM}, 2022.

\bibitem[BGL{\etalchar{+}}12]{bansal2012lp}
Nikhil Bansal, Anupam Gupta, Jian Li, Juli{\'a}n Mestre, Viswanath Nagarajan, and Atri Rudra.
\newblock When lp is the cure for your matching woes: Improved bounds for stochastic matchings.
\newblock {\em Algorithmica}, 63(4):733--762, 2012.

\bibitem[BHK{\etalchar{+}}24]{banihashem2024power}
Kiarash Banihashem, MohammadTaghi Hajiaghayi, Dariusz~R Kowalski, Piotr Krysta, and Jan Olkowski.
\newblock Power of posted-price mechanisms for prophet inequalities.
\newblock In {\em Proceedings of the 2024 Annual ACM-SIAM Symposium on Discrete Algorithms (SODA)}, pages 4580--4604. SIAM, 2024.

\bibitem[BK10]{bahmani2010improved}
Bahman Bahmani and Michael Kapralov.
\newblock Improved bounds for online stochastic matching.
\newblock In Mark de~Berg and Ulrich Meyer, editors, {\em Algorithms - {ESA} 2010, 18th Annual European Symposium, Liverpool, UK, September 6-8, 2010. Proceedings, Part {I}}, volume 6346 of {\em Lecture Notes in Computer Science}, pages 170--181. Springer, 2010.

\bibitem[BK23]{braun2023simplepricing}
Alexander Braun and Thomas Kesselheim.
\newblock Simplified prophet inequalities for combinatorial auctions.
\newblock In {\em 2023 Symposium on Simplicity in Algorithms (SOSA)}, pages 381--389, 2023.

\bibitem[BM19]{baek2019prophet}
Jackie Baek and Will Ma.
\newblock Prophet inequalities on the intersection of a matroid and a graph.
\newblock {\em CoRR}, abs/1906.04899, 2019.

\bibitem[BMR20]{borodin2020bipartite}
Allan Borodin, Calum MacRury, and Akash Rakheja.
\newblock Bipartite stochastic matching: Online, random order, and iid models.
\newblock {\em arXiv preprint arXiv:2004.14304}, 2020.

\bibitem[BSSX16]{brubach2016new}
Brian Brubach, Karthik~Abinav Sankararaman, Aravind Srinivasan, and Pan Xu.
\newblock New algorithms, better bounds, and a novel model for online stochastic matching.
\newblock In {\em Proceedings of the 24th Annual European Symposium on Algorithms (ESA)}, pages 24:1--24:16, 2016.

\bibitem[BSSX20]{brubach2020online}
Brian Brubach, Karthik~Abinav Sankararaman, Aravind Srinivasan, and Pan Xu.
\newblock Online stochastic matching: New algorithms and bounds.
\newblock {\em Algorithmica}, 82(10):2737--2783, 2020.

\bibitem[CC23]{correa2023subadditive}
Jos\'{e} Correa and Andr\'{e}s Cristi.
\newblock A constant factor prophet inequality for online combinatorial auctions.
\newblock In {\em Proceedings of the 55th Annual ACM Symposium on Theory of Computing}, STOC 2023, page 686–697, New York, NY, USA, 2023. Association for Computing Machinery.

\bibitem[CCF{\etalchar{+}}22]{correa2022pricing}
Jos{\'e} Correa, Andr{\'e}s Cristi, Andr{\'e}s Fielbaum, Tristan Pollner, and S.~Matthew Weinberg.
\newblock Optimal item pricing in online combinatorial auctions.
\newblock In Karen Aardal and Laura Sanit{\`a}, editors, {\em Integer Programming and Combinatorial Optimization}, pages 126--139, Cham, 2022. Springer International Publishing.

\bibitem[CGKM20]{chawla2020matroid}
Shuchi Chawla, Kira Goldner, Anna~R. Karlin, and J.~Benjamin Miller.
\newblock Non-adaptive matroid prophet inequalities.
\newblock {\em CoRR}, abs/2011.09406, 2020.

\bibitem[CHMS10]{chawla2010multi}
Shuchi Chawla, Jason~D Hartline, David~L Malec, and Balasubramanian Sivan.
\newblock Multi-parameter mechanism design and sequential posted pricing.
\newblock In {\em Proceedings of the 42nd Annual ACM Symposium on Theory of Computing (STOC)}, pages 311--320, 2010.

\bibitem[CIK{\etalchar{+}}09]{chen2009approximating}
Ning Chen, Nicole Immorlica, Anna~R Karlin, Mohammad Mahdian, and Atri Rudra.
\newblock Approximating matches made in heaven.
\newblock In {\em Proceedings of the 36th International Colloquium on Automata, Languages and Programming (ICALP)}, pages 266--278, 2009.

\bibitem[DFKL20]{dutting2020prophet}
Paul D\"utting, Michal Feldman, Thomas Kesselheim, and Brendan Lucier.
\newblock Prophet inequalities made easy: Stochastic optimization by pricing nonstochastic inputs.
\newblock {\em SIAM Journal on Computing (SICOMP)}, 49(3), 2020.

\bibitem[DGR{\etalchar{+}}23]{dutting2023prophet}
Paul D{\"{u}}tting, Evangelia Gergatsouli, Rojin Rezvan, Yifeng Teng, and Alexandros Tsigonias{-}Dimitriadis.
\newblock Prophet secretary against the online optimal.
\newblock In Kevin Leyton{-}Brown, Jason~D. Hartline, and Larry Samuelson, editors, {\em Proceedings of the 24th {ACM} Conference on Economics and Computation, {EC} 2023, London, United Kingdom, July 9-12, 2023}, pages 561--581. {ACM}, 2023.

\bibitem[DK15]{duetting2015polymatroid}
Paul D{\"{u}}tting and Robert Kleinberg.
\newblock Polymatroid prophet inequalities.
\newblock In Nikhil Bansal and Irene Finocchi, editors, {\em Algorithms - {ESA} 2015 - 23rd Annual European Symposium, Patras, Greece, September 14-16, 2015, Proceedings}, volume 9294 of {\em Lecture Notes in Computer Science}, pages 437--449. Springer, 2015.

\bibitem[DKL20]{dutting2020log}
Paul D{\"u}tting, Thomas Kesselheim, and Brendan Lucier.
\newblock An ${O}(\log \log m)$ prophet inequality for subadditive combinatorial auctions.
\newblock In {\em 2020 IEEE 61st Annual Symposium on Foundations of Computer Science (FOCS)}, pages 306--317. IEEE, 2020.

\bibitem[DSSX21]{dickerson2021reusable}
John~P. Dickerson, Karthik~A. Sankararaman, Aravind Srinivasan, and Pan Xu.
\newblock Allocation problems in ride-sharing platforms: Online matching with offline reusable resources.
\newblock {\em ACM Trans. Econ. Comput.}, 9(3), June 2021.

\bibitem[EFGT20]{ezra2020online}
Tomer Ezra, Michal Feldman, Nick Gravin, and Zhihao~Gavin Tang.
\newblock Online stochastic max-weight matching: prophet inequality for vertex and edge arrival models.
\newblock In {\em Proceedings of the 21st ACM Conference on Economics and Computation (EC)}, pages 769--787, 2020.

\bibitem[FGL15]{feldman2015combinatorial}
Michal Feldman, Nick Gravin, and Brendan Lucier.
\newblock Combinatorial auctions via posted prices.
\newblock In {\em Proceedings of the 26th Annual ACM-SIAM Symposium on Discrete Algorithms (SODA)}, pages 123--135, 2015.

\bibitem[FHTZ20]{fahrbach2020edge}
Matthew Fahrbach, Zhiyi Huang, Runzhou Tao, and Morteza Zadimoghaddam.
\newblock Edge-weighted online bipartite matching.
\newblock In {\em Proceedings of the 61st Symposium on Foundations of Computer Science (FOCS)}, 2020.
\newblock To Appear.

\bibitem[FLT{\etalchar{+}}22]{fu2022oblivious}
Hu~Fu, Pinyan Lu, Zhihao~Gavin Tang, Abner Turkieltaub, Hongxun Wu, Jinzhao Wu, and Qianfan Zhang.
\newblock Oblivious online contention resolution schemes.
\newblock In {\em Symposium on Simplicity in Algorithms (SOSA)}, pages 268--278, 2022.

\bibitem[FMMM09]{feldman2009online}
Jon Feldman, Aranyak Mehta, Vahab Mirrokni, and S~Muthukrishnan.
\newblock Online stochastic matching: Beating 1-1/e.
\newblock In {\em Proceedings of the 50th Symposium on Foundations of Computer Science (FOCS)}, pages 117--126, 2009.

\bibitem[FNS19]{feng2019reusable}
Yiding Feng, Rad Niazadeh, and Amin Saberi.
\newblock Linear programming based online policies for real-time assortment of reusable resources.
\newblock {\em SSRN Electronic Journal}, 01 2019.

\bibitem[FNS22]{feng2022reusable}
Yiding Feng, Rad Niazadeh, and Amin Saberi.
\newblock Near-optimal bayesian online assortment of reusable resources.
\newblock In {\em Proceedings of the 23rd ACM Conference on Economics and Computation}, EC '22, page 964–965, New York, NY, USA, 2022. Association for Computing Machinery.

\bibitem[FSZ16]{feldman2016online}
Moran Feldman, Ola Svensson, and Rico Zenklusen.
\newblock Online contention resolution schemes.
\newblock In {\em Proceedings of the 27th Annual ACM-SIAM Symposium on Discrete Algorithms (SODA)}, pages 1014--1033, 2016.

\bibitem[GHH{\etalchar{+}}21]{gao2021improved}
Ruiquan Gao, Zhongtian He, Zhiyi Huang, Zipei Nie, Bijun Yuan, and Yan Zhong.
\newblock Improved online correlated selection.
\newblock In {\em Proceedings of the 62nd Symposium on Foundations of Computer Science (FOCS)}, 2021.
\newblock To Appear.

\bibitem[GHK{\etalchar{+}}14]{gobel2014online}
Oliver G{\"{o}}bel, Martin Hoefer, Thomas Kesselheim, Thomas Schleiden, and Berthold V{\"{o}}cking.
\newblock Online independent set beyond the worst-case: Secretaries, prophets, and periods.
\newblock In Javier Esparza, Pierre Fraigniaud, Thore Husfeldt, and Elias Koutsoupias, editors, {\em Automata, Languages, and Programming - 41st International Colloquium, {ICALP} 2014, Copenhagen, Denmark, July 8-11, 2014, Proceedings, Part {II}}, volume 8573 of {\em Lecture Notes in Computer Science}, pages 508--519. Springer, 2014.

\bibitem[GKPS06]{gandhi2006dependent}
Rajiv Gandhi, Samir Khuller, Srinivasan Parthasarathy, and Aravind Srinivasan.
\newblock Dependent rounding and its applications to approximation algorithms.
\newblock {\em Journal of the ACM (JACM)}, 53(3):324--360, 2006.

\bibitem[GKS19]{gamlath2019beating}
Buddhima Gamlath, Sagar Kale, and Ola Svensson.
\newblock Beating greedy for stochastic bipartite matching.
\newblock In {\em Proceedings of the Thirtieth Annual ACM-SIAM Symposium on Discrete Algorithms}, pages 2841--2854. SIAM, 2019.

\bibitem[GU23]{goyal2023onlinestochastic}
Vineet Goyal and Rajan Udwani.
\newblock Online matching with stochastic rewards: Optimal competitive ratio via path-based formulation.
\newblock {\em Oper. Res.}, 71(2):563--580, 2023.

\bibitem[GW19]{gravin2019prophet}
Nikolai Gravin and Hongao Wang.
\newblock Prophet inequality for bipartite matching: Merits of being simple and non adaptive.
\newblock In {\em Proceedings of the 20th ACM Conference on Economics and Computation (EC)}, pages 93--109, 2019.

\bibitem[HJS{\etalchar{+}}23]{huang2023onlinestochastic}
Zhiyi Huang, Hanrui Jiang, Aocheng Shen, Junkai Song, Zhiang Wu, and Qiankun Zhang.
\newblock Online matching with stochastic rewards: Advanced analyses using configuration linear programs.
\newblock In Jugal Garg, Max Klimm, and Yuqing Kong, editors, {\em Web and Internet Economics - 19th International Conference, {WINE} 2023, Shanghai, China, December 4-8, 2023, Proceedings}, volume 14413 of {\em Lecture Notes in Computer Science}, pages 384--401. Springer, 2023.

\bibitem[HKS07]{hajiaghayi2007automated}
Mohammad~Taghi Hajiaghayi, Robert Kleinberg, and Tuomas Sandholm.
\newblock Automated online mechanism design and prophet inequalities.
\newblock In {\em Proceedings of the 22nd AAAI Conference on Artificial Intelligence (AAAI)}, pages 58--65, 2007.

\bibitem[HMZ11]{haeupler2011online}
Bernhard Haeupler, Vahab~S. Mirrokni, and Morteza Zadimoghaddam.
\newblock Online stochastic weighted matching: Improved approximation algorithms.
\newblock In Ning Chen, Edith Elkind, and Elias Koutsoupias, editors, {\em Internet and Network Economics - 7th International Workshop, {WINE} 2011, Singapore, December 11-14, 2011. Proceedings}, volume 7090 of {\em Lecture Notes in Computer Science}, pages 170--181. Springer, 2011.

\bibitem[HS21]{huang2021online}
Zhiyi Huang and Xinkai Shu.
\newblock Online stochastic matching, poisson arrivals, and the natural linear program.
\newblock In Samir Khuller and Virginia~Vassilevska Williams, editors, {\em {STOC} '21: 53rd Annual {ACM} {SIGACT} Symposium on Theory of Computing, Virtual Event, Italy, June 21-25, 2021}, pages 682--693. {ACM}, 2021.

\bibitem[HSY22]{huang2022power}
Zhiyi Huang, Xinkai Shu, and Shuyi Yan.
\newblock The power of multiple choices in online stochastic matching.
\newblock In Stefano Leonardi and Anupam Gupta, editors, {\em {STOC} '22: 54th Annual {ACM} {SIGACT} Symposium on Theory of Computing, Rome, Italy, June 20 - 24, 2022}, pages 91--103. {ACM}, 2022.

\bibitem[HZ20]{huang2020onlinestochastic}
Zhiyi Huang and Qiankun Zhang.
\newblock Online primal dual meets online matching with stochastic rewards: configuration lp to the rescue.
\newblock In {\em Proceedings of the 52nd Annual ACM SIGACT Symposium on Theory of Computing}, STOC 2020, page 1153–1164, New York, NY, USA, 2020. Association for Computing Machinery.

\bibitem[JL13]{jaillet2013online}
Patrick Jaillet and Xin Lu.
\newblock Online stochastic matching: New algorithms with better bounds.
\newblock {\em Mathematics of Operations Research}, 2013.

\bibitem[JMZ22]{jiang2022prophet}
Jiashuo Jiang, Will Ma, and Jiawei Zhang.
\newblock Tight guarantees for multi-unit prophet inequalities and online stochastic knapsack.
\newblock In {\em Proceedings of the 2022 Annual ACM-SIAM Symposium on Discrete Algorithms (SODA)}, pages 1221--1246, 2022.

\bibitem[KS78]{krengel1978semiamarts}
Ulrich Krengel and Louis Sucheston.
\newblock On semiamarts, amarts, and processes with finite value.
\newblock {\em Probability on Banach spaces}, 4:197--266, 1978.

\bibitem[KVV90]{karp1990optimal}
Richard~M Karp, Umesh~V Vazirani, and Vijay~V Vazirani.
\newblock An optimal algorithm for on-line bipartite matching.
\newblock In {\em Proceedings of the 22nd Annual ACM Symposium on Theory of Computing (STOC)}, pages 352--358, 1990.

\bibitem[KW19]{kleinberg2019matroid}
Robert Kleinberg and S~Matthew Weinberg.
\newblock Matroid prophet inequalities and applications to multi-dimensional mechanism design.
\newblock {\em Games and Economic Behavior}, 113:97--115, 2019.

\bibitem[LS18]{lee2018optimal}
Euiwoong Lee and Sahil Singla.
\newblock Optimal online contention resolution schemes via ex-ante prophet inequalities.
\newblock In {\em Proceedings of the 26th Annual European Symposium on Algorithms (ESA)}, pages 57:1--57:14, 2018.

\bibitem[Luc17]{lucier2017economic}
Brendan Lucier.
\newblock An economic view of prophet inequalities.
\newblock {\em ACM SIGecom Exchanges}, 16(1):24--47, 2017.

\bibitem[Meh13]{mehta2013online}
Aranyak Mehta.
\newblock Online matching and ad allocation.
\newblock {\em Foundations and Trends{\textregistered} in Theoretical Computer Science}, 8(4):265--368, 2013.

\bibitem[MGS12]{manshadi2012online}
Vahideh~H Manshadi, Shayan~Oveis Gharan, and Amin Saberi.
\newblock Online stochastic matching: Online actions based on offline statistics.
\newblock {\em Mathematics of Operations Research}, 37(4):559--573, 2012.

\bibitem[MMG23]{macrury2023ocrs}
Calum MacRury, Will Ma, and Nathaniel Grammel.
\newblock On (random-order) online contention resolution schemes for the matching polytope of (bipartite) graphs.
\newblock In {\em Proceedings of the 2023 Annual ACM-SIAM Symposium on Discrete Algorithms (SODA)}, pages 1995--2014, 2023.

\bibitem[MP12]{mehta2012online}
Aranyak Mehta and Debmalya Panigrahi.
\newblock Online matching with stochastic rewards.
\newblock In {\em Symposium on Foundations of Computer Science (FOCS)}, 2012.

\bibitem[MWZ15]{mehta2015onlinestochastic}
Aranyak Mehta, Bo~Waggoner, and Morteza Zadimoghaddam.
\newblock Online stochastic matching with unequal probabilities.
\newblock In {\em Proceedings of the Twenty-Sixth Annual ACM-SIAM Symposium on Discrete Algorithms (SODA-15)}, pages 1388--1404, 2015.

\bibitem[NSW23]{naor2023dependentroundingarxiv}
Joseph Naor, Aravind Srinivasan, and David Wajc.
\newblock Online dependent rounding schemes.
\newblock {\em CoRR}, abs/2301.08680, 2023.

\bibitem[PPSW21]{papadimitriou2021online}
Christos Papadimitriou, Tristan Pollner, Amin Saberi, and David Wajc.
\newblock Online stochastic max-weight bipartite matching: Beyond prophet inequalities.
\newblock In {\em Proceedings of the 22nd ACM Conference on Economics and Computation (EC)}, pages 763--764, 2021.

\bibitem[PRSW22]{pollner2022ocrs}
Tristan Pollner, Mohammad Roghani, Amin Saberi, and David Wajc.
\newblock Improved online contention resolution for matchings and applications to the gig economy.
\newblock In {\em Proceedings of the 23rd ACM Conference on Economics and Computation}, EC '22, page 321–322, New York, NY, USA, 2022. Association for Computing Machinery.

\bibitem[Rub16]{rubinstein2016prophet}
Aviad Rubinstein.
\newblock Beyond matroids: secretary problem and prophet inequality with general constraints.
\newblock In {\em Proceedings of the Forty-Eighth Annual ACM Symposium on Theory of Computing}, STOC '16, page 324–332, New York, NY, USA, 2016. Association for Computing Machinery.

\bibitem[SC84]{samuel1984comparison}
Ester Samuel-Cahn.
\newblock Comparison of threshold stop rules and maximum for independent nonnegative random variables.
\newblock {\em the Annals of Probability}, 12(4):1213--1216, 1984.

\bibitem[Sri01]{srinivasan2001distributions}
Aravind Srinivasan.
\newblock Distributions on level-sets with applications to approximation algorithms.
\newblock In {\em Proceedings of the 42nd Symposium on Foundations of Computer Science (FOCS)}, pages 588--597, 2001.

\bibitem[SW21]{saberi2021greedy}
Amin Saberi and David Wajc.
\newblock The greedy algorithm is \emph{not} optimal for on-line edge coloring.
\newblock In {\em Proceedings of the 48th International Colloquium on Automata, Languages and Programming (ICALP)}, pages 109:1--109:18, 2021.

\bibitem[TT22]{torrico2022dynamic}
Alfredo Torrico and Alejandro Toriello.
\newblock Dynamic relaxations for online bipartite matching.
\newblock {\em INFORMS Journal on Computing}, 2022.

\bibitem[TWW22]{tang2022fractional}
Zhihao~Gavin Tang, Jinzhao Wu, and Hongxun Wu.
\newblock ({F}ractional) online stochastic matching via fine-grained offline statistics.
\newblock In {\em Proceedings of the 54th Annual ACM Symposium on Theory of Computing (STOC)}, pages 77--90, 2022.

\end{thebibliography}
\newpage
\appendix
\section{Informative Examples and Observations} \label{app:informative-examples}

In this section, we give some examples and observations which might help to gain a deeper understanding of the problem. 

\subsection{The Generalization of \texorpdfstring{\cite{braverman2022max}}{} Fails}

Given the attention previously dedicated to the unit-capacity case, we first ask how these algorithms perform for the capacitated problem. Previous works for matching have all used the LP relaxation \eqref{LP} with $c_t = 1$, in the special case where each success probability $q_{i,t}$ equals 1. In the simplest case where every resource $t$ either (i) arrives with a fixed capacity and values, with probability $p_t$ or (ii) does not arrive, with probability $1-p_t$, the algorithm works in the following way: in the case that resource $t$ arrives, every available user $i$ sends a proposal to $t$ with probability $$\frac{ x_{i,t} }{ p_t \cdot \left( 1 - \sum_{t' < t} x_{i, t'} \right) },$$ an expression that is at most 1 by \eqref{LP} Constraint~\eqref{eqn:PPSWConstraint}. The resource is matched to the proposing user with highest value $v_{i,t}$. \cite{braverman2022max} show that this algorithm gives a $(1-1/e)$-approximation against \eqref{LP}, and hence also the optimum online benchmark. 

To account for capacities, we might naturally generalize this algorithm to match an arriving resource $t$ to the top $c_t$ proposing users. Surprisingly, this small modification drastically changes the algorithm's performance. 

\bdmzeroapprox*

\begin{proof}

Take some $n$ such that $n > \frac{2}{\epsilon}$, and consider an instance with $n$ users and two resources. The first resource has a capacity of $n$ (i.e., values are additive over all users), arrives with probability $1-1/n$ and values are $1$ for each user individually. The second resource is unit-capacity, arrives with probability 1, and values are $n^2$ for each user individually. All allocations are successful with probability 1. 

The unique optimal solution to \eqref{LP} sets $x_{i,1} = 1-1/n$ for every pair $(i,1)$ incident to the first resource, and sets $x_{i,2} = 1/n$ for every pair $(i,2)$ incident to the second resource. Thus, when running (the natural generalization of) \cite{braverman2022max}, every user proposes to the first resource if it arrives, and hence with probability $1-1/n$ all users are assigned in the first timestep. If the first resource does not arrive, exactly one user is allocated to the second unit-capacity resource. Hence the expected gain of the algorithm is $\left( 1 - \frac{1}{n} \right) \cdot n + \frac{1}{n} \cdot n^2 = 2n-1$. However clearly for this instance $\opton \ge n^2$. 
\end{proof}

\subsection{Positive Correlation is Required}

Next, we argue that we \emph{need} to have positive correlation for general capacitated resource allocation.

\positivecorrelationrequired*

\begin{proof}
Let $F_{i,t}$ denote an indicator for user $i$ being free just before the arrival of resource $t$. 
Consider resource $t$ with capacity two arriving with probability $\epsilon$ which is adjacent to two users $\{i,j\}$ with unit values. Imagine the LP sets a value of $\epsilon$ on each edge. To achieve an approximation factor of $(0.5 + \kappa)$ against LP, we are required to have that that the expected number of users assigned to $t$
is at least $(0.5 + \kappa) \cdot 2\epsilon$. Equivalently, we must have $$\Pr[F_{i,t+1}] + \Pr[F_{j,t+1}] < 2- (0.5 + \kappa) \cdot 2\epsilon$$ implying $$\Pr[F_{i,t+1}] \cdot \Pr[F_{j,t+1}] < (1 - (0.5 + \kappa) \cdot \epsilon)^2 = 1 - (1 + 2 \kappa) \epsilon + O(\epsilon^2).$$ However, because $i$ and $j$ can only be matched if $t$ arrives, we have 
\begin{align*}
\Pr[F_{i,t+1} \wedge F_{j,t+1}] & \ge 1- \epsilon > \Pr[F_{i,t+1}] \cdot \Pr[F_{j,t+1}]   ,
\end{align*} where the final inequality holds for sufficiently small $\epsilon$. 
\end{proof}

\subsection{On the Gap of \texorpdfstring{\eqref{LP}}{}}

\begin{exm}
    There exists an instance of online capacitated allocation where $$\frac{\opton}{\textup{OPT}\eqref{LP}} \le 0.75.$$
\end{exm}
\begin{proof}

Consider an instance with two offline users, and two stochastic arrivals. The first resource has capacity 2, and arrives with probability $\nicefrac{1}{2}$; the second resource has capacity 1 and arrives with probability 1. Both resources have a value of 1 for each user; every edge is successful with probability 1.  

The optimum online algorithm achieves a value of 2 if the first user arrives, and a value of 1 otherwise, hence achieving $1.5$ in expectation. However, a feasible solution to \eqref{LP} sets $x_{i,t} = \nicefrac{1}{2}$ for every edge $(i,t)$, hence achieving a value of 2. 
\end{proof}

\subsection[A Bound Depending on min c\_t]{A Bound Depending on $\min_t c_t$.}
\label{sec:bound_depending_kt}

As mentioned in \Cref{sec:approx_late_edges}, the bound following \Cref{eqn:bound_conditional_At} in the proof of \Cref{lemma:expected_At_with_condition} is not tight if all $c_t$ are strictly greater than one. Still, even though this step looks quite lossy at first glance, we are not losing much in our analysis by replacing $\min_t c_t$ with one. To see this, consider 
replacing the last inequality in the proof of \Cref{lemma:expected_At_with_condition} with a bound depending on $\min_t c_t$. Doing so, we get 
\begin{align}\label{eqn:bound_depending_ct}
	\E[A_t \mid t \text{ arrived}, F_{i,t}] &\le 1 - \tau  + \Deltac \cdot (0.5 + \kappa) \cdot c_t \nonumber \\
	&= \left( \frac{1 - \tau}{c_t}  + \Deltac \cdot (0.5 + \kappa) \right) \cdot c_t \nonumber \\ 
	& \le \left( \frac{1 - \tau}{\min_{t'} c_{t'}}  + \Deltac \cdot (0.5 + \kappa) \right) \cdot c_t  .
\end{align}

As a consequence, in order to show the desired lower bound on $\rho_{i,t}$, we first can use the same reasoning as we used in order to derive \Cref{eqn:rho_lower_bound}, but use Inequality~\eqref{eqn:bound_depending_ct} instead:
\begin{align*}
    \rho_{i,t} \ge (1 - (0.5 + \kappa) \cdot y_{i,t}) \cdot \left( \tau - \left( \frac{1-\tau}{\min_{t'} c_{t'}} + \Deltac \cdot(0.5 + \kappa) \right) \right)  .
\end{align*}
Thus, the right-hand side needs to be at least as large as $(0.5 + \kappa) y_{i,t} - (0.5 - \kappa)$. In other words, we are required to show that
\begin{align*}
    (1 - (0.5 + \kappa) \cdot y_{i,t}) \cdot \left( \tau - \left( \frac{1-\tau}{\min_{t'} c_{t'}} + \Deltac \cdot(0.5 + \kappa) \right) \right) \geq (0.5 + \kappa) y_{i,t} - (0.5 - \kappa)  .
\end{align*}
Hence we can take any $\kappa$ such that
\begin{align}\label{eqn:inequality_for_c}
    \tau - \left( \frac{1-\tau}{\min_{t'} c_{t'}} + \Deltac \cdot(0.5 + \kappa) \right) \geq \frac{2 \kappa}{0.5-\kappa}  .
\end{align}

As a consequence, we can now solve \Cref{eqn:inequality_for_c} for $\kappa$ in order to improve upon the constant of $\approxconstant$ which we used initially, as a function of $\min_t c_t$. In \Cref{tab:improved_c_values_depending_kt}, we state these constants for $\min_t c_t \in \{ 2,\dots,9 \} $, demonstrating that there is little loss in our analysis of \Cref{allocationalgexact} when replacing $\min_t c_t$ with $1$. 

\begin{table}[h]
    \centering
    \begin{tabular}{c|c|c|c|c|c|c|c|c|c }
        \textbf{$\min_t c_t$} & 1 & 2 & 3 & 4 & 5 & 6 & 7 & 8 & 9 \\  \hline
        \textbf{$\kappa$} & 0.0115 & 0.0126 & 0.0131 & 0.0133 & 0.0134 & 0.0135 & 0.01362 & 0.01367 & 0.01371 \\
    \end{tabular}
    \vspace{0.3em}
    \caption{Values of $\kappa$ depending on $\min_t c_t$}
    \label{tab:improved_c_values_depending_kt}
\end{table}

\section{Deferred Proofs}
\label{app:deferredproofs}

In this section, we provide proofs which were deferred from the main body. 

\subsection{Proof of \Cref{thm:beat_half_for_comb_auc}}
\label{app:beat_half_comb_auc}

\beathalfforcombauc*
\begin{proof}
    We apply Theorem 19 of \cite{banihashem2024power}, as our problem of capacitated resource allocation can be viewed exactly as what they call a \emph{prophet inequalities problem}. Using their notation, we take $\mathcal{A}^{\text{inp}}$ to be \Cref{allocationalggeneral}, with expected social welfare $\E [v(\mathcal{A}^{\text{inp}})]$. Note that \Cref{allocationalggeneral} is what \cite{banihashem2024power} call ``past-valuation-independent,'' as its allocation decision for buyer $t$ depends only on the set of available items, the arriving valuation/capacity $v_t(\cdot)$, and the LP solution calculated from knowledge of the input distributions. Note also that for each buyer $t$, the outcome space (what \cite{banihashem2024power} refer to as ``$X_t$'') is of size at most $\binom{n}{c_t} = \text{poly}(n)$ because $c_t$ is upper bounded by a constant. Finally, although our distribution over $v_t(\cdot)$ is not continuous, it is not hard to satisfy this assumption by adding a small amount of noise or a tiebreaking coordinate (as mentioned in \cite{banihashem2024power}). 

    Hence, there is a pricing based algorithm $\mathcal{A}^{\text{out}}$ which uses $\text{poly}(T, \binom{n}{\max_t c_t}, \nicefrac{1}{\epsilon})$ many samples, runs in time $\text{poly}(T, \binom{n}{\max_t c_t}, \nicefrac{1}{\epsilon})$ and whose expected social welfare satisfies $$  \E [\mathcal{A}^{\text{out}}] \geq (1- \epsilon) \cdot \E [\mathcal{A}^{\text{in}}].$$ 
\end{proof}

\subsection{Proof of \texorpdfstring{\Cref{observation:LP_relax_OPT_on}}{}}
\label{app:LP_relax_OPT}

\observationlprelaxopton*

\begin{proof}
    Define an indicator random variable $X_{i,t}$ for every pair $(i,t)$, which is one if and only if the optimum online algorithm allocates user $i$ to resource $t$. In addition, let $Q_{i,t}$ be the indicator which is one if the assignment of the pair $(i,t)$ was successful; i.e. the independent Bernoulli coin flip with probability $q_{i,t}$ comes up heads.
    
    Denote by $x^\ast_{i,t} = \E[X_{i,t}]$. First, note that the welfare achieved by the optimum online algorithm is $$\opton = \E \left[ \sum_{i,t} v_{i,t} X_{i,t} Q_{i,t} \right] = \sum_{i,t} v_{i,t} \cdot x^\ast_{i,t} \cdot q_{i,t},$$ coinciding with the objective of \eqref{LP}. Here the expectation is over the randomness in $X_{i,t}$ as well as the success probabilities for $(i,t)$, and we crucially use that the successful realization of $(i,t)$ is independent of our decision to allocate along $(i,t)$. 
    
    Also, observe that for any resource $t$, we have $\sum_i X_{i,t} = 0$ if the resource does not arrive, and $\sum_i X_{i,t} \leq c_t$ if the resource arrives, as any algorithm is allowed to allocate at most $c_t$ users to resource $t$ if the resource arrives. Hence $$\sum_i x^\ast_{i,t}   = \E \left[ \sum_i X_{i,t} \right] = \Pr \left[ t \text{ arrives} \right] \cdot \E \left[ \sum_i X_{i,t} \growingmid t \text{ arrives} \right] \leq p_t \cdot c_t.$$ 

    Finally, note that if resource $t$ arrives, the optimum online algorithm can only allocate user $i$ if it is available. For user $i$ being available, it had not to be allocated to some previous resource $t' < t$ whose independent coin flip $Q_{i,t'}$ was successful as well. 
    Crucially, for any online algorithm, the event that user $i$ is available at time $t$ is independent of the arrival of resource $t$ (this does not hold for an offline algorithm). Hence, we observe 
    \begin{align*}
        x^\ast_{i,t} & = \E [ X_{i,t} ]  = \Pr \left[ t \text{ arrives} \right] \cdot \E \left[ X_{i,t} \growingmid t \text{ arrives} \right] \\ & \leq p_t \cdot \E \left[ 1 - \sum_{t' < t} X_{i,t'} Q_{i,t'} \growingmid t \text{ arrives} \right] \\ & = p_t \cdot \E \left[ 1 - \sum_{t' < t} X_{i,t'} Q_{i,t'} \right]  =  p_t \cdot \left( 1 - \sum_{t' < t} x^\ast_{i,t'} \cdot q_{i,t'} \right). 
    \end{align*}

    As a consequence, $\{ x_{i,t}^\ast \}_{i,t}$ is a feasible solution to \eqref{LP} and hence, $\text{OPT}\eqref{LP} \geq \opton$.
\end{proof}

\subsection{Proof of \texorpdfstring{\Cref{claim:inductive_step}}{}}
\label{app:proof_inductive_step}

\claiminductivestep*

\begin{proof}
    Plugging in the definition of $f(z) = 1 + z \cdot \left( \frac{(0.5+\kappa)^2}{1 - z \cdot (0.5 + \kappa)} \right)$, the claim is equivalent to 
    \begin{align*}
        \left( 1 + \frac{(0.5+\kappa)^2 y}{1 - (0.5+\kappa)y} \right) \left( 1 + \frac{(0.5+\kappa)^2 x }{1 - (0.5+\kappa) (x+y)} \right) \le 1 + \frac{(0.5+\kappa)^2 (x+y)}{1 - (0.5+\kappa)(x+y)} .
    \end{align*}
    Multiplying out the left-hand side and subtracting $1 + \frac{(0.5+\kappa)^2 x }{1 - (0.5+\kappa) (x+y)}$ on both sides, this is equivalent to
    \begin{align*}
        \frac{(0.5+\kappa)^2 y}{1 - (0.5+\kappa)y} + \frac{(0.5+\kappa)^2 y}{1 - (0.5+\kappa)y} \cdot \frac{(0.5+\kappa)^2 x }{1 - (0.5+\kappa) (x+y)} \le \frac{(0.5+\kappa)^2 y}{1 - (0.5+\kappa)(x+y)} .
    \end{align*}
    If $y = 0$, the claim is trivially true. If $y > 0$, we can divide both sides by $(0.5+\kappa)^2 y$ to get 
    \begin{align*}
        \frac{1}{1 - (0.5+\kappa)y} + \frac{1}{1 - (0.5+\kappa)y} \cdot \frac{(0.5+\kappa)^2 x }{1 - (0.5+\kappa) (x+y)} \le \frac{1}{1 - (0.5+\kappa)(x+y)} .
    \end{align*}
    Multiplying both sides by $(1 - (0.5+\kappa)y)\cdot (1 - (0.5+\kappa) (x+y))$, we get
    \begin{align*}
        1 - (0.5+\kappa) (x+y) + (0.5+\kappa)^2 x \le 1 - (0.5+\kappa)y .
    \end{align*}
    Subtracting $1 - (0.5+\kappa)y$ on both sides, we finally end up with
    \begin{align*}
        - (0.5+\kappa) x + (0.5+\kappa)^2 x \le 0 
    \end{align*}
    which is clear. 
\end{proof}

\section{Beyond Bernoulli Distributions}
\label{app:beyond_bernoulli}

When not restricting the model to Bernoulli arrivals, for every round $t$, there is a known distribution $\{p_{t,j}\}_j$ over valuation vectors $\{v_{i,t,j}\}_{i}$ and a capacity $c_{t,j}$. Upon the arrival of resource $t$, it samples one index $j \in \{1,\dots,m\}$ with probability $p_{t,j}$\footnote{We assume without loss of generality that all resource share the same space of valuation vectors and capacities, and we can set $p_{t,j} = 0$ if realization $j$ is not feasible for resource $t$. Also, we assume that resources always arrive by adding a valuation vector containing only zeros with the probability of resource $t$ not arriving.}, and realizes capacity $c_{t,j}$ and values $\{ v_{i,t,j} \}_i$ over users. For the ease of exposition, we discuss general arrivals in the case that each success probability $q_{i,t,j} = 1$, and describe the changes needed to handle arbitrary success probabilities $q_{i,t,j} \in [0,1]$ in \Cref{app:stochasticrewards}.

\paragraph{Generalized LP} We generalize \ref{LP} as follows. 
\begin{align}
	\nonumber  \max \  &  \sum_{i,t,j} x_{i,t,j} \cdot v_{i,t,j} && \tag{General-LP\textsubscript{on}} \label{generalLP} \\
    \text{s.t. }& \sum_t \sum_j x_{i,t,j} \leq 1 && \text{for all } i \in I \label{eqn:generalalloconce} \\
	&\sum_{i} x_{i,t,j} \le p_{t,j} \cdot c_{t,j} && \text{for all } t \in [T], j \in [m] \\
	& 0 \le x_{i,t,j} \le p_{t,j} \cdot \left( 1 - \sum_{t' < t} \sum_{j'} x_{i, t', j'} \right) && \text{for all } i \in I, t \in [T], j \in [m]\label{eqn:generalPPSWConstraint}
 \end{align}

In an equivalent manner to \Cref{observation:LP_relax_OPT_on}, we can argue that also for general distributions, $\OPT\eqref{generalLP} \geq \opton$, i.e. \ref{generalLP} is a relaxation of the optimum online algorithm. 

\paragraph{Generalized Algorithm.} In order to round any fractional LP solution to an integral one in an online fashion, we extend our \Cref{allocationalgexact} as follows: In round $t$, we see the realization of index $j$. We replace all previous LP variables with the ones from the generalized LP for index $j$ and run the slightly modified \Cref{allocationalggeneral}.

\begin{algorithm}[H]
	\caption{}
	\label{allocationalggeneral}
	\begin{algorithmic}[1]
		\State $\kappa \gets \approxconstant$ 
		\State Solve \eqref{generalLP} for $\{x_{i,t,j} \}$ 
		\For{each time $t$}
		\State Observe index $j$ sampled from $(p_{t,j})_j$
		\State Define users $\FPtj := \textsf{PS}( ( x_{i,t,j}/p_{t,j} )_{i \in I})$ 
		\For{each user $i \in \FPtj$}
		\If{$i$ is available}
		\State Allocate $i$ to $t$ with probability $\alpha_{i,t} := \min \left(1, \frac{0.5 + \kappa }{1 - (0.5 + \kappa) \cdot \sum_{t' < t} \sum_{j'} x_{i,t',j'}} \right)$
		\EndIf 
		
		\EndFor
		
		\State Let $A_{t,j} \gets \text{number of users allocated to } t \text{ with sampled index } j \text{ thus far}$ \label{line:sample_defAtj_general}
		\State Define users $\SPtj:= \textsf{PS} \left( \left( \left(1 - \frac{A_{t,j}}{c_{t,j}} \right) \cdot x_{i,t,j}/p_{t,j} \right)_{i \in I} \right)$ 
		\For{each user $i \in \SPtj$ with $\alpha_{i,t} = 1$}
		\If{$i$ is available} 
		\State Compute $\rho_{i,t,j} := \E \left[ \mathbbm{1}[i \text{ available after \Cref{line:sample_defAtj_general}}] \cdot \left( 1 - \frac{A_{t,j}}{c_{t,j}} \right) \mid t \text{ sampled index } j \right] $ \label{line:beyondberncomputerho}
        \State $\beta_{i,t,j} \gets \min \Big(1, \left( (0.5 + \kappa) \cdot \sum_{t' < t} \sum_j x_{i,t',j}- (0.5 - \kappa) \right) \cdot \frac{1}{\rho_{i,t,j}} \Big).$
		\State Allocate $i$ to $t$ with prob. ${\beta_{i,t,j}}$
		\EndIf
		\EndFor
		
		\EndFor
	\end{algorithmic}
\end{algorithm}	

As in our Bernoulli case, observe that we choose  ${\beta_{i,t,j}}$ in a way so that the following holds: $\Pr[(i,t) \text{ assigned for sampled index } j ] = (0.5+\kappa) \cdot x_{i,t,j}$. 
Also, note that this algorithm can be implemented in polynomial time in the number of resources and users and the size of the support of the distributions. Concerning the computation of $\rho_{i,t,j}$, we can observe that for our choice of $\kappa = \approxconstant$, the generalized analysis also shows that any $\rho_{i,t,j}$ is lower bounded by a constant; equivalently to the Bernoullli case. This can be used to estimate $\rho_{i,t,j}$ via samples with a multiplicative error as small as desired, implying a $(0.5 + \kappa - \epsilon)$-approximate algorithm, following the logic of \Cref{app:sample_based_algo}.

\paragraph{Generalized Analysis.} 
In order to prove the generalization of \Cref{theorem:main_theorem}, the major work is to change the syntax of the lemmas on the way. We do not give details for all lemmas but rather provide the key steps on what to change and how to overcome obstacles on the way.

First, we extend and change several definitions such as $y_{i,t} := \sum_{t' < t} \sum_j x_{i,t',j}$ or $\mathcal{A}_1^j$, $\mathcal{A}_2^j$ as the set of assignments $(i,t)$ if the realized index is $j$ via a first or second proposal. The lemmas, observations and statements which referred to ``$t$ arriving" are now with respect to the event ``$t$ realizes index $j$". For example, when talking about assigning $i$ to $t$ via a first proposal, we replace this by saying that we assign $i$ to $t$ via a first proposal when $t$ realized the valuation vector with index $j$. 

The proofs for the analysis of early pairs directly carry over after adapting the syntax. For late pairs, the generalization of the proof of \Cref{lemma:approx_late_edges} (i) is also straightforward, as is the combination of both analyses at the end. 

We need to take some care in generalizing the proof of \Cref{lemma:approx_late_edges} (ii). The majority of the steps can be extended straightforwardly via syntactic generalization from \Cref{sec:analysis} (or \Cref{app:sample_based_algo} with an estimate of the expectation in \Cref{line:beyondberncomputerho}). In contrast, the proof of generalized versions of the correlation bound from \Cref{subsection:correlation_bound}, and in particular \Cref{claim:claim_for_corrbound} need some short updates. Note however that as \Cref{claim:claim_for_corrbound} only concerns early pairs, it is not affected by the updates for a sample-based algorithm as in \Cref{app:sample_based_algo}.

To see why \Cref{claim:claim_for_corrbound} also holds in the more general variant, we go through its proof steps one-by-one. Concerning the generalization of Step~\ref{induction:step_2} we note that the probability of both users being free after time $t+1$ can still be decomposed as the product of the probability of both being free before times the conditional probability of assigning neither via a first proposal (as in \Cref{equation:induction_step}). Still, we are required to sum the latter conditional probabilities for all possible realizations of $j$. 
Doing so, we first follow Steps~\ref{induction:step_2} and \ref{induction:step_3} from the Bernoulli case.
During Step~\ref{induction:step_4}, we need to show that for two distinct users $i, i'$ and resource $t$, the following inequality holds: 
\begin{align}\label{eqn:inequality_general_distributions}
    &\alpha_{i,t} \alpha_{i',t} \Pr[F_{i,t}] \Pr[F_{i',t}] \left( \sum_j \frac{x_{i,t,j} x_{i',t,j}}{p_{t,j}} - \left( \sum_j x_{i,t,j} \right) \left( \sum_j x_{i',t,j} \right) \right) \\ & \nonumber \hspace{0.5cm } \leq \Pr[F_{i,t}] \Pr[F_{i',t}]  (0.5 + \kappa) \left(1 - \alpha_{i',t} \sum_j x_{i',t,j} \right) \alpha_{i,t} \left( \sum_j x_{i,t,j} \right) .
\end{align}
In order to argue that this inequality is indeed true, we depart from the proof of the Bernoulli case by controlling the term $\sum_j \frac{x_{i,t,j} x_{i',t,j}}{p_{t,j}}$ via the online constraint for the user $i'$. By Constraint~\eqref{eqn:generalPPSWConstraint}, we know that $$\frac{ x_{i',t,j}}{p_{t,j}} \leq 1 - y_{i',t} .$$ Using this, we can bound $$\sum_j \frac{x_{i,t,j} x_{i',t,j}}{p_{t,j}} \le \left( 1 - y_{i',t} \right) \sum_j x_{i,t,j} . $$ Plugging this into the left-hand side of \Cref{eqn:inequality_general_distributions} and rearranging terms, we can conclude in a similar way as we did using \Cref{fact:bound_alpha_via_online} in the Bernoulli case. Afterwards, Step~\ref{induction:step_5} of the correlation bound can again proceed via syntactic generalization which concludes the proof for general distributions.

\section{Stochastic Rewards} \label{app:stochasticrewards}

In \Cref{sec:analysis} we assumed for convenience that every pair $(i,t)$ had a success probability $q_{i,t} = 1$. This was mainly for convenience of notation, as the guarantees for our algorithm carry over to the case of arbitrary success probabilities $q_{i,t} \in [0,1]$. The changes can furthermore be adapted to our sample-based algorithm (as in \Cref{app:sample_based_algo}) and algorithm for non-Bernoulli arrivals (as in \Cref{app:beyond_bernoulli}), although for simplicity we start by extending the algorithm for Bernoulli arrivals without samples. 

We recall that we say $i$ is \emph{allocated} to $t$ if it is one of the at most $c_t$ items which we attempt to assign to $t$, and we say it is \emph{successfully allocated} to $t$ if and only if it is allocated and the independent success indicator $\text{Ber}(q_{i,t})$ comes up heads. Note that if for every $(i,t)$ we have that $i$ is allocated to $t$ with probability $(0.5 + \kappa) \cdot x_{i,t}$, then because of the independence of the success indicators we have that the expected welfare contribution of $(i,t)$ is $(0.5 + \kappa ) \cdot x_{i,t} \cdot q_{i,t}$ and hence we achieve a $(0.5 + \kappa)$-approximation to \eqref{LP}. 

If we naturally update our definition of $y_{i,t} := \sum_{t' < t} x_{i,t'} \cdot q_{i,t'}$ (instead of $\sum_{t' < t} x_{i,t'}$), many of the changes required to the analysis are syntactic. We inductively show that the probability $(i,t)$ is allocated is $(0.5 + \kappa) \cdot x_{i,t}$, and hence have as part of the inductive hypothesis that $\Pr[F_{i,t}] = 1 - (0.5 + \kappa) \cdot y_{i,t}$. Thus, the probability an early $(i,t)$ is allocated is precisely $$p_t \cdot \Pr[F_{i,t}] \cdot \frac{0.5 + \kappa}{1 - (0.5 + \kappa) y_{i,t}} = (0.5 + \kappa ) \cdot x_{i,t}.$$ The analysis for late pairs also generalizes syntactically, with the caveat that we must take care to consider how the independent $\text{Ber}(q_{i,t})$ affect the correlation bound of \Cref{corollary:bound_correlation_unrestricted}. Intuitively, as these Bernoullis are independent of our proposals and history, they should not contribute to worse positive correlation. This is formalized below. 

We first consider the proof of \Cref{claim:claim_for_corrbound}. Our original proof (the grey line below) used the bound \color{gray}
$$\Pr[F_{i,t+1} \wedge F_{j,t+1}] \le \Pr[F_{i,t} \wedge F_{j,t}] \cdot \left( 1 - x_{i,t} \cdot \alpha_{i,t} - x_{j,t} \cdot \alpha_{j,t} + \frac{x_{i,t} \cdot x_{j,t}}{p_t}\cdot \alpha_{i,t} \cdot \alpha_{j,t} \right).$$ \color{black}
In the new setting, with the independence of successful matches, we have instead 
$$\Pr[F_{i,t+1} \wedge F_{j,t+1}] \le \Pr[F_{i,t} \wedge F_{j,t}] \cdot \left( 1 - x_{i,t}  \alpha_{i,t}  q_{i,t} - x_{j,t}  \alpha_{j,t}  q_{j,t} + \frac{x_{i,t} q_{i,t} \cdot x_{j,t} q_{j,t} }{p_t}\cdot \alpha_{i,t} \cdot \alpha_{j,t} \right).$$ Hence, we will define $\tilde{x}_{i,t} := x_{i,t} \cdot q_{i,t}$ and $\tilde{x}_{j,t} := x_{j,t} \cdot q_{j,t}$. As $\Pr[F_{i,t+1}]/\Pr[F_{i,t}] =  1 - x_{i,t} \alpha_{i,t} q_{i,t}  =  1 - \tilde{x}_{i,t} \cdot \alpha_{i,t} $ the proof proceeds identically with this syntactic change, and implies $$\Pr[F_{i,t+1} \wedge F_{j,t+1}] \le \Pr[F_{i,t+1}] \cdot \Pr[F_{j,t+1}] \cdot f(y_{i,t} + \tilde{x}_{i,t}) = \Pr[F_{i,t+1}] \cdot \Pr[F_{j,t+1}] \cdot f(y_{i,t+1}).$$ With this change in place the proof of \Cref{corollary:bound_correlation_unrestricted} can be modified syntatically with the new definition of $y_{i,t}$. Indeed, the only property we need is that is that $A_i$ should now denote the event that $i$ is \emph{successfully} allocated to an arrival in $[t^i, t-1]$ (and similarly for $j$). Then,we compute $$\Pr[A_i] = \sum_{t' \in [t^i, t-1]} (0.5 + \kappa) \cdot x_{i,t'} \cdot q_{i,t'} \le 2 \kappa $$ where we use the updated definition of $y_{i,t}$. 

We also can readily integrate these changes in our (sampling-based) algorithm for arrivals from general distributions. In particular, we have the following LP relaxation and algorithm. 

\begin{align}
	\nonumber  \max \  &  \sum_{i,t,j} x_{i,t,j} \cdot q_{i,t,j} \cdot v_{i,t,j} && \tag{General-LP\textsubscript{on}-Stochastic} \label{generalLPstochastic} \\
    \text{s.t. } & \sum_t \sum_j x_{i,t,j} \cdot q_{i,t,j} \leq 1 && \text{for all } i \in I \label{eqn:generalalloconcestoch} \\
	&\sum_{i} x_{i,t,j} \le p_{t,j} \cdot c_{t,j} && \text{for all } t \in [T], j \in [m] \\
	& 0 \le x_{i,t,j} \le p_{t,j} \cdot \left( 1 - \sum_{t' < t} \sum_{j'} x_{i, t', j'} \cdot q_{i, t', j'} \right) && \text{for all } i \in I, t \in [T], j \in [m]
 \end{align}

\begin{algorithm}[H]
	\caption{}
	\label{allocationalggeneralstochastic}
	\begin{algorithmic}[1]
		\State $\kappa \gets \approxconstant$ 
		\State Solve \eqref{generalLPstochastic} for $\{x_{i,t,j} \}$ 
		\For{each time $t$}
		\State Observe index $j$ sampled from $(p_{t,j})_j$
		\State Define users $\FPtj := \textsf{PS}( ( x_{i,t,j}/p_{t,j} )_{i \in I})$ 
		\For{each user $i \in \FPtj$}
		\If{$i$ is available}
		\State Allocate $i$ to $t$ with probability $\alpha_{i,t} := \min \left(1, \frac{0.5 + \kappa }{1 - (0.5 + \kappa) \cdot \sum_{t' < t} \sum_{j'} x_{i,t',j'} \cdot q_{i,t',j'}} \right)$
		\EndIf 
		
		\EndFor
		
		\State Let $A_{t,j} \gets \text{number of users allocated to } t \text{ with sampled index } j \text{ thus far}$ \label{line:sample_defAtj_stochastic}
		\State Define users $\SPtj:= \textsf{PS} \left( \left( \left(1 - \frac{A_{t,j}}{c_{t,j}} \right) \cdot x_{i,t,j}/p_{t,j} \right)_{i \in I} \right)$ 
		\For{each user $i \in \SPtj$ with $\alpha_{i,t} = 1$}
		\If{$i$ is available} 
		\State Compute $\rho_{i,t,j} := \E \left[ \mathbbm{1}[i \text{ available after \Cref{line:sample_defAtj_stochastic}}] \cdot \left( 1 - \frac{A_{t,j}}{c_{t,j}} \right) \mid t \text{ sampled index } j \right] $ 
        \State $\beta_{i,t,j} \gets \min \Big(1, \left( (0.5 + \kappa) \cdot \sum_{t' < t} \sum_j x_{i,t',j} \cdot q_{i,t',j} - (0.5 - \kappa) \right) \cdot \frac{1}{\rho_{i,t,j}} \Big).$
		\State Allocate $i$ to $t$ with prob. ${\beta_{i,t,j}}$
		\EndIf
		\EndFor
		
		\EndFor
	\end{algorithmic}

\end{algorithm}	

To analyze the algorithm, we can now generalize $y_{i,t} := \sum_{t' < t} \sum_{j'} x_{i,t',j'} \cdot q_{i, t', j'}$, so that $\alpha_{i,t} = \min \left(1, \frac{0.5 + \kappa }{1 - (0.5 + \kappa) \cdot  y_{i,t} } \right)$ and $\beta_{i,t,j} = \min \Big(1, \left( (0.5 + \kappa) \cdot y_{i,t} - (0.5 - \kappa) \right) \cdot \frac{1}{\rho_{i,t,j}} \Big)$. Similarly, we can define $\tilde{x}_{i,t,j} := x_{i,t,j} \cdot q_{i,t,j}$. Using $y_{i,t}$ and $\tilde{x}_{i,t,j}$, the arguments of \Cref{app:beyond_bernoulli} now generalize syntatically, as described above for the Bernoulli case. The stochastic rewards do not change the argument from \Cref{app:beyond_bernoulli} that $\rho_{i,t,j}$ is bounded away from $0$ by a constant, and hence can be computed efficiently within a multiplicative error factor when running the polynomial-time sample-based algorithm. 

\end{document}